%% file: hyperbolic_vertex_cover_approximation.tex
\documentclass[a4paper,UKenglish,cleveref, autoref]{lipics-v2021}

\usepackage[utf8]{inputenc}
\usepackage{amsmath}
\usepackage{amssymb}
\usepackage{amsthm}
\usepackage{commath}
\usepackage{booktabs}
\usepackage{graphicx}
\usepackage{longtable}
\usepackage{mathtools}
\usepackage{pifont}
\usepackage{siunitx}
\usepackage{timetravel}
\usepackage{wasysym}
\usepackage[textsize=tiny]{todonotes}
\usepackage{bbm}

\makeatletter
\newcommand{\vast}{\bBigg@{4}}
\newcommand{\Vast}{\bBigg@{5}}
\makeatother

\setCounters{theorem}

\DeclareMathOperator{\dist}{dist}
\DeclareMathOperator{\acosh}{acosh}

\newcommand{\Spacious}[1]{\, #1 \,}
\newcommand{\SpaciousCdot}{\Spacious{\cdot}}

\title{Efficiently Approximating Vertex Cover on Scale-Free Networks
  with Underlying Hyperbolic Geometry} 

\titlerunning{Approximating Vertex Cover on Networks with Underlying
  Hyperbolic
  Geometry}

\author{Thomas Bläsius}{Karlsruhe Institute of Technology\\{Karlsruhe, Germany}}{thomas.blaesius@kit.edu}{https://orcid.org/0000-0003-2450-744X}{}

\author{Tobias Friedrich}{Hasso Plattner Institute, University of Potsdam\\{Potsdam, Germany}}{tobias.friedrich@hpi.de}{https://orcid.org/0000-0003-0076-6308}{}

\author{Maximilian Katzmann}{Hasso Plattner Institute, University of Potsdam\\{Potsdam, Germany}}{maximilian.katzmann@hpi.de}{https://orcid.org/0000-0002-9302-5527}{}

\authorrunning{T. Bläsius, T. Friedrich, M. Katzmann}

\Copyright{Thomas Bläsius, Tobias Friedrich, Maximilian Katzmann}



\nolinenumbers 

\hideLIPIcs  

\EventEditors{Markus Bl\"{a}ser and Benjamin Monmege}
\EventNoEds{2}
\EventLongTitle{38th International Symposium on Theoretical Aspects of Computer Science (STACS 2020)}
\EventShortTitle{STACS 2021}
\EventAcronym{STACS}
\EventYear{2021}
\EventDate{March 16--19, 2021}
\EventLocation{Saarbr\"{u}cken, Germany}
\EventLogo{}
\SeriesVolume{42}
\ArticleNo{23}

\begin{document}
\maketitle

\begin{abstract}
  Finding a minimum vertex cover in a network is a fundamental
  NP-complete graph problem.  One way to deal with its computational
  hardness, is to trade the qualitative performance of an algorithm
  (allowing non-optimal outputs) for an improved running time.  For
  the vertex cover problem, there is a gap between theory and practice
  when it comes to understanding this tradeoff.  On the one hand, it
  is known that it is NP-hard to approximate a minimum vertex cover
  within a factor of $\sqrt{2}$.  On the other hand, a simple greedy
  algorithm yields close to optimal approximations in practice.

  A promising approach towards understanding this discrepancy is to
  recognize the differences between theoretical worst-case instances
  and real-world networks.  Following this direction, we close the gap
  between theory and practice by providing an algorithm that
  efficiently computes close to optimal vertex cover approximations on
  hyperbolic random graphs; a network model that closely resembles
  real-world networks in terms of degree distribution, clustering, and
  the small-world property.  More precisely, our algorithm computes a
  $(1 + o(1))$-approximation, asymptotically almost surely, and has a
  running time of $\mathcal{O}(m \log(n))$.

  The proposed algorithm is an adaption of the successful greedy
  approach, enhanced with a procedure that improves on parts of the
  graph where greedy is not optimal.  This makes it possible to
  introduce a parameter that can be used to tune the tradeoff between
  approximation performance and running time.  Our empirical
  evaluation on real-world networks shows that this allows for
  improving over the near-optimal results of the greedy approach.
\end{abstract}

\newpage

\section{Introduction}
\label{sec:introduction}

A \emph{vertex cover} of a graph is a subset of the vertices that
leaves the graph edgeless upon deletion.  Since the problem of finding
a smallest vertex cover is NP-complete~\cite{k-racp-72}, there are
probably no algorithms that solve it efficiently.  Nevertheless, the
problem is relevant due to its applications in computational
biology~\cite{acf-kavcp-04}, scheduling~\cite{elw-vcms-16}, and
internet security~\cite{ffg-cowpun-07}.
Therefore, there is an ongoing effort in exploring methods that can be
used in practice~\cite{ai-bf-16, acl-ichvc-12}, and while they often
work well, they still cannot guarantee efficient running times.

A commonly used approach to overcoming this issue are approximation
algorithms.  There, the idea is to settle for a near-optimal solution
while guaranteeing an efficient running time.  For the vertex cover
problem, a simple greedy approach computes an approximation in
quasi-linear time by iteratively adding the vertex with the largest
degree to the cover and removing it from the graph.  In general
graphs, this algorithm, which we call \emph{standard greedy}, cannot
guarantee a better approximation ratio than $\log(n)$, i.e., there are
graphs where it produces a vertex cover whose size exceeds the one of
an optimum by a factor of $\log(n)$~\cite{j-aacp-74}.  This can be
improved to a $2$-approximation using a simple linear-time algorithm.
The best known polynomial time approximation reduces the factor to
$2 - \Theta(\log(n)^{-1/2})$~\cite{k-barvcp-09}.  However, assuming
the unique games conjecture, it is NP-hard to approximate an optimal
vertex cover within a factor of $2 - \varepsilon$ for all
$\varepsilon > 0$~\cite{kr-v-08} and it is proven that finding a
$\sqrt{2}$-approximation is NP-hard~\cite{sms-psggh-18}.

Therefore, it is rather surprising that the standard greedy algorithm
not only beats the $2$-approximation on autonomous systems graphs like
the internet~\cite{pw-ialscs-05}, it also performs well on many
real-world networks, obtaining approximation ratios that are very
close to $1$~\cite{dgd-vccn-13}.  This leaves a gap between the
theoretical worst-case bounds and what is observed in practice.  One
approach to explaining this discrepancy is to consider the differences
between the examined instances.  Theoretical bounds are often obtained
by designing worst-case instances.  However, real-world networks
rarely resemble the worst case.  More realistic statements can be
obtained by making assumptions about the solution
space~\cite{bl-asie-12, crv-sris-17}, or by restricting the analysis
to networks with properties that are observed in the real world.

Many real networks, like social networks, communication networks, or
protein-interaction networks, are considered to be
\emph{scale-free}~\cite{asvw-h-20, n-sfcn-03, scm-t-21}.  Such graphs
feature a power-law degree distribution (only few vertices have high
degree, while many vertices have low degree), high clustering (two
vertices are likely to be adjacent if they have a common neighbor),
and a small diameter.

Previous efforts to obtain more realistic insights into the
approximability of the vertex cover problem have focused on networks
that feature only one of these properties, namely a power-law degree
distribution~\cite{cfr-ggdsfn-16, gh-a-14, vs-mvcgr-16}.  With this
approach, guarantees for the approximation factor of the standard
greedy algorithm were improved to a constant, compared to $\log(n)$ on
general graphs~\cite{cfr-ggdsfn-16}.  Moreover, it was shown that it
is possible to compute an expected $(2 - \varepsilon)$-approximation
for a constant $\varepsilon$, in polynomial time on such
networks~\cite{gh-a-14} and this was later improved to about $1.7$
depending on properties of the distribution~\cite{vs-mvcgr-16}.
However, it was also shown that even on graphs that have a power-law
degree distribution, the vertex cover problem remains NP-hard to
approximate within some constant factor~\cite{cfr-ggdsfn-16}.
This indicates, that focusing on networks that only feature a
power-law degree distribution, is not sufficient to explain why vertex
cover can be approximated so well in practice.

The goal of this paper is to narrow this gap between theory and
practice, by considering a random graph model that features all of the
three mentioned properties of scale-free networks.  The
\emph{hyperbolic random graph model} was introduced by Krioukov et
al.~\cite{kpk-h-10} and it was shown that the graphs generated by the
model have a power-law degree distribution and high
clustering~\cite{gpp-rhg-12, fhms-c-21}, as well as a small
diameter~\cite{ms-k-19}.  Consequently, they are good representations
of many real-world networks~\cite{bpk-sihm-10, gbas-whhgit-16,
  vs-mehssn-14}.  Additionally, the model is conceptually simple,
making it accessible to mathematical analysis.  With it we have
previously derived a theoretical explanation for why the bidirectional
breadth-first search works well in practice~\cite{bff-espsf-22}.
Moreover, we have shown that the vertex cover problem can be solved
exactly in polynomial time on hyperbolic random graphs, with high
probability~\cite{bffk-svcpt-20}.  However, we note that the degree of
the polynomial is unknown and on large networks even quadratic
algorithms are not efficient enough to obtain results in reasonable
time.

In this paper, we link the success of the standard greedy approach to
structural properties of hyperbolic random graphs, identify the parts
of the graph where it does not behave optimally, and use these
insights to derive a new approximation algorithm.  On hyperbolic
random graphs, this algorithm achieves an approximation ratio of
$1 + o(1)$, asymptotically almost surely (i.e., with probability
$1 - o(1)$), and maintains an efficient running time of
$\mathcal{O}(m \log(n))$, where $n$ and $m$ denote the number of
vertices and edges in the graph, respectively.  Since the average
degree of hyperbolic random graphs is constant, with high
probability~\cite{k-girggcg-18}, this implies a quasi-linear running
time on such networks.  Moreover, we introduce a parameter that can be
used to tune the trade-off between approximation quality and running
time of the algorithm, facilitating an improvement over the standard
greedy approach.  While our algorithm depends on the coordinates of
the vertices in the hyperbolic plane, we propose an adaptation of it
that is oblivious to the underlying geometry (only relying on the
adjacency information of the graph) and compare its approximation
performance to the standard greedy algorithm on a selection of
real-world networks.  On average our algorithm reduces the error of
the standard greedy approach to less than $50\%$.

The paper is structured as follows.  We first give an overview of our
notation and preliminaries in \Cref{sec:preliminaries} and derive a
new approximation algorithm based on prior insights about vertex cover
on hyperbolic random graphs in \Cref{sec:algorithm}.  Afterwards, we
analyze its approximation ratio in \Cref{sec:approximation} and
evaluate its performance empirically in \Cref{sec:experiments}.

\section{Preliminaries}
\label{sec:preliminaries}

Let $G = (V, E)$ be an undirected and connected graph.  We denote the
number of vertices and edges in~$G$ with $n$ and $m$, respectively.
The number of vertices in a set $S \subseteq V$ is denoted by $|S|$.
The \emph{neighborhood} of a vertex $v$ is defined as
$N(v) = \{ w \in V \mid \{ v, w \} \in E \}$.  The size of the
neighborhood, called the \emph{degree} of $v$, is denoted by
$\deg(v) = \vert N(v) \vert$.  For a subset $S \subseteq V$, we use
$G[S]$ to denote the \emph{induced subgraph} of $G$ obtained by
removing all vertices in $V \setminus S$.

\subparagraph{The Hyperbolic Plane.}

After choosing a designated origin $O$ in the two-dimensional
hyperbolic plane, together with a reference ray starting at $O$, a
point $p$ is uniquely identified by its \emph{radius} $r(p)$, denoting
the hyperbolic distance to $O$, and its \emph{angle} (or \emph{angular
  coordinate})~$\varphi(p)$, denoting the angular distance between the
reference ray and the line through $p$ and $O$.  The hyperbolic
distance between two points $p$ and $q$ is given by
\begin{align}
  \dist(p, q) = \acosh(\cosh(r(p))\cosh(r(q)) - \sinh(r(p))\sinh(r(q))\cos(\Delta_\varphi(p, q))), \notag
\end{align}
where $\cosh(x) = (e^x + e^{-x}) / 2$,
$\sinh(x) = (e^x - e^{-x}) / 2$, and
\begin{align*}
  \Delta_\varphi(p, q) = \pi - \vert \pi - \vert \varphi(p) - \varphi(q) \vert\vert
\end{align*}
denotes the angular distance between $p$ and $q$.  If not stated
otherwise, we assume that computations on angles are performed modulo
$2\pi$.

In the hyperbolic plane a disk of radius $r$ has an area of
$2\pi(\cosh(r) - 1)$.  That is, the area grows exponentially with the
radius.  In contrast, this growth is polynomial in Euclidean space.

\subparagraph{Hyperbolic Random Graphs.}

Hyperbolic random graphs are obtained by distributing $n$ points
independently and uniformly at random within a disk of radius $R$ and
connecting any two of them if and only if their hyperbolic distance is
at most $R$.  See \Cref{fig:graph} for an example.  The disk radius
$R$ (which matches the connection threshold) is given by
$R = 2\log(n) + C$, where the constant $C \in \mathbb{R}$ depends on
the average degree of the network, as well as the power-law exponent
$\beta = 2\alpha + 1$, with $\alpha \in (1/2, 1)$, which are also
assumed to be constants.  The coordinates of the vertices are drawn as
follows. For vertex~$v$ the angular coordinate, denoted
by~$\varphi(v)$, is drawn uniformly at random from $[0, 2\pi)$ and the
radius of $v$, denoted by $r(v)$, is sampled according to the
probability density function
\begin{align*}
  f(r) = \frac{\alpha\sinh(\alpha r)}{\cosh(\alpha R) - 1},
\end{align*}
for $r \in [0, R]$.  For $r > R$, $f(r) = 0$.  Then,
\begin{align}
  \label{eq:probability-density-function}
  f(r, \varphi) &= \frac{1}{2\pi} \frac{\alpha \sinh(\alpha r)}{\cosh(\alpha R) - 1} \\ \notag
                &= \frac{\alpha}{2 \pi} e^{-\alpha(R - r)}(1 + \Theta(e^{-\alpha R} - e^{-2\alpha r}))
\end{align}
is their joint distribution function.

In the chosen regime for $\alpha$ the resulting graphs have a giant
component of size $\Omega(n)$~\cite{bfm-lchmcn-15}, while all other
components have at most polylogarithmic size~\cite[Corollary
13]{fk-dhrg-18}, with high probability.  Throughout the paper, we
refer only to the giant component when addressing hyperbolic random
graphs.

We denote areas in the hyperbolic disk with calligraphic capital
letters.  The set of vertices in an area $\mathcal{A}$ is denoted by
$V(\mathcal{A})$.  The probability for a given vertex to lie in
$\mathcal{A}$ is given by its measure
$\mu(\mathcal{A}) = \iint_\mathcal{A} f(r, \varphi) \dif\varphi \dif
r$.  The hyperbolic distance between two vertices $u$ and $v$
increases with increasing angular distance between them.  The maximum
angular distance such that they are still connected by an edge is
bounded by~\cite[Lemma 3.2]{k-shrg-16}
\begin{align}
  \label{eq:maximum-angular-distance}
  \theta(r(u), r(v)) &= \arccos\left( \frac{\cosh(r(u))\cosh(r(v)) - \cosh(R)}{\sinh(r(u))\sinh(r(v))} \right) \notag \\
                     &= 2e^{(R - r(u) - r(v))/2}(1 \pm \Theta(e^{R - r(u) - r(v)})).
\end{align}

\subparagraph{Hyperbolic Random Graphs with an Expected Number of
  Vertices.}

We are often interested in the probability that one or more vertices
fall into a certain area of the hyperbolic disk during the sampling
process of a hyperbolic random graph.  Computing such a probability
becomes significantly harder, once the positions of some vertices are
already known, since that introduces stochastic dependencies.  For
example, if all $n$ vertices are sampled into an area $\mathcal{A}$,
the probability for a vertex to lie outside of $\mathcal{A}$ is $0$.
In order to overcome such issues, we use an approach (that has been
often used on hyperbolic random graphs before, see for
example~\cite{fk-dhrg-18, km-bdrhg-15}), where the vertex positions in
the hyperbolic disk are sampled using an inhomogeneous Poisson point
process.  For a given number of vertices $n$, we refer to the
resulting model as \emph{hyperbolic random graphs with $n$ vertices in
  expectation}.  After analyzing properties of this simpler model, we
can translate the results back to the original model, by conditioning
on the fact that the resulting distribution is equivalent to the one
originally used for hyperbolic random graphs.  More formally, this can
be done as follows.

A hyperbolic random graph with $n$ vertices in expectation is obtained
using an inhomogeneous Poisson point process to distribute the
vertices in the hyperbolic disk.  In order to get $n$ vertices in
expectation, the corresponding intensity function $f_{P}(r, \varphi)$
at a point $(r, \varphi)$ in the disk is chosen as
\begin{align*}
  f_{P}(r, \varphi) = e^{(R - C)/2} f(r, \varphi),
\end{align*}
where $f(r, \varphi)$ is the original probability density function
used to sample hyperbolic random graphs (see
\Cref{eq:probability-density-function}). Let $P$ denote the set of
random variables representing the points produced by this process.
Then $P$ has two properties.  First, the number of vertices in $P$
that are sampled into two disjoint areas are independent random
variables.  Second, the expected number of points in $P$ that fall
within an area $\mathcal{A}$ is given by
\begin{align*}
  \iint_{\mathcal{A}} f_{P}(r, \varphi) \dif r \dif \varphi &= n \iint_{\mathcal{A}} f(r, \varphi) \dif r \dif \varphi = n \mu(\mathcal{A}).
\end{align*}
By the choice of $f_{P}$ the number of vertices sampled into the disk
matches $n$ only in expectation, i.e.,
$\mathbb{E}[\vert P \vert] = n$.  However, we can now recover the
original distribution of the vertices, by conditioning on the fact
that $\vert P \vert = n$, as shown in the following lemma.
Intuitively, it states that probabilistic statements on hyperbolic
random graphs with $n$ vertices in expectation can be translated to
the original hyperbolic random graph model by taking a small penalty
in certainty.  We note that proofs of how to bound this penalty have
been sketched before~\cite{fk-dhrg-18, km-bdrhg-15}.  For the sake of
completeness, we give an explicit proof.  In the following, we use
$G_P$ to denote a hyperbolic random graph with $n$ vertices in
expectation and point set $P$.  Moreover, we use $\mathbf{P}$ to
denote a property of a graph and for a given graph $G$ we denote the
event that $G$ has property $\mathbf{P}$ with $E(G, \mathbf{P})$.

\begin{lemma}
  \label{lem:hrg-in-expectation}
  Let $G_P$ be a hyperbolic random graph with $n$ vertices in
  expectation, let~$\mathbf{P}$~be a property, and let $c > 0$ be a
  constant, such that $\Pr[E(G_P, \mathbf{P})] = \mathcal{O}(n^{-c})$.
  Then, for a hyperbolic random graph $G'$ with $n$ vertices it holds
  that
  \begin{align*}
    \Pr[E(G', \mathbf{P})] = \mathcal{O}(n^{-c + 1/2}).
  \end{align*}
\end{lemma}
\begin{proof}
  The probability that $G'$ has property $\mathbf{P}$ can be obtained
  by taking the probability that a hyperbolic random graph $G_P$ with
  $n$ vertices in expectation has it, and conditioning on the fact
  that \emph{exactly} $n$ vertices are produced during its sampling
  process.  That is,
  \begin{align*}
    \Pr[E(G', \mathbf{P})] = \Pr[E(G_P, \mathbf{P}) \mid \vert P \vert = n].
  \end{align*}
  This probability can now be computed using the definition for
  conditional probabilities, i.e.,
  \begin{align*}
    \Pr[E(G_P, \mathbf{P}) \mid \vert P \vert = n] = \frac{\Pr[E(G_P, \mathbf{P}) \land \vert P \vert = n]}{\Pr[\vert P \vert = n]},
  \end{align*}
  where the $\land$-operator denotes that both events occur.  For the
  numerator, we have $\Pr[E(G_P, \mathbf{P})] = \mathcal{O}(n^{-c})$
  by assumption.  Constraining this to events where
  $\vert P \vert = n$ cannot increase the probability and we obtain
  $\Pr[E(G_P, \mathbf{P}) \land \vert P \vert = n] =
  \mathcal{O}(n^{-c})$.  For the denominator, recall that
  $\vert P \vert$ is a random variable that follows a Poisson
  distribution with mean $n$.  Therefore, we have
  \begin{align*}
    \Pr[\vert P \vert = n] = \frac{e^{-n}n^n}{n!} = \Theta(n^{-1/2}).
  \end{align*}
  The quotient can, thus, be bounded by
  \begin{equation*}
    \Pr[E(G', \mathbf{P})] = \frac{\mathcal{O}(n^{-c})}{\Theta(n^{-1/2})} = \mathcal{O}(n^{-c + 1/2}). 
  \end{equation*}
\end{proof}

\subparagraph{Probabilities.}

Since we are analyzing a random graph model, our results are of
probabilistic nature.  To obtain meaningful statements, we show that
they hold \emph{with high probability} (with probability $1 -
\mathcal{O}(n^{-1})$), or \emph{asymptotically almost surely} (with
probability $1 - o(1)$).  The following Chernoff bound can be used to
show that certain events occur with high probability.

\begin{theorem}[{Chernoff Bound~\cite[Theorems 4.4 and 4.5]{mu-pc-05}}]
  \label{thm:chernoff}
  Let $X_1, \dots, X_n$ be independent random variables with $X_i \in
  \{0, 1\}$ and let $X$ be their sum.  Then, for $\varepsilon \in (0, 1]$
  \begin{align*}
    \Pr[X \ge (1 + \varepsilon)\mathbb{E}[X]] &\le e^{- \varepsilon^2/3 \SpaciousCdot \mathbb{E}[X]}.
  \end{align*}
\end{theorem}

Usually, it suffices to show that a random variable does not exceed an
upper bound.  The following corollary shows that a bound on the
expected value suffices to obtain concentration.

\begin{corollary}
  \label{col:chernoff}
  Let $X_1, \dots, X_n$ be independent random variables with $X_i \in
  \{0, 1\}$, let $X$ be their sum, and let $f(n)$ be an upper bound on
  $\mathbb{E}[X]$.  Then, for $\varepsilon \in (0, 1)$ 
  \begin{align*}
    \Pr[X \ge (1 + \varepsilon)f(n)] \le e^{-\varepsilon^2/3 \SpaciousCdot f(n)}.
  \end{align*}
\end{corollary}
\begin{proof}
  We define random variables $X_1', \dots, X_n'$ with $X_i' \ge X_i$
  for every outcome, in such a way that $X' = \sum_{i \in [n]} X_i'$
  has expected value $\mathbb{E}[X'] = f(n)$.  Note that $X' \ge X$
  for every outcome and that $X'$ exists as $f(n) \ge \mathbb{E}[X]$.
  Since $X \le X'$, it holds that
  \begin{align*}
    \Pr[X \ge (1 + \varepsilon)f(n)] \le \Pr[X' \ge (1 + \varepsilon)f(n)] = \Pr[X' \ge (1 + \varepsilon)\mathbb{E}(X')].
  \end{align*}
  Using \Cref{thm:chernoff} we can derive that
  \begin{equation*}
    \Pr[X' \ge (1 + \varepsilon)\mathbb{E}[X']] \le e^{-\varepsilon^2/3 \SpaciousCdot \mathbb{E}[X']} = e^{-{\varepsilon^2/3 \SpaciousCdot f(n)}}.
  \end{equation*}
\end{proof}

While the Chernoff bound considers the sum of indicator random
variables, we often have to deal with different functions of random
variables.  In this case tight bounds on the probability that the
function deviates a lot from its expected value can be obtained using
the method of bounded differences.  Let $X_1, \dots, X_n$ be
independent random variables taking values in a set $S$.  We say that
a function $f \colon S^{n} \rightarrow \mathbb{R}$ satisfies the
\emph{bounded differences condition} if for all $i \in [n]$ there
exists a $\Delta_i \ge 0$ such that
\begin{align}
  \label{eq:bounded-differences-condition}
  \vert f(\boldsymbol{x}) - f(\boldsymbol{x}') \vert \le \Delta_i,
\end{align}
for all $\boldsymbol{x}, \boldsymbol{x}' \in S^{n}$ that differ only
in the $i$-th component.

\begin{theorem}[{Method of Bounded Differences~\cite[Corollary
    5.2]{dp-cmara-12}}]
  \label{thm:bounded-differences}
  Let $X_1, \dots, X_n$ be independent random variables taking values
  in a set $S$ and let $f \colon S^{n} \rightarrow \mathbb{R}$ be a
  function that satisfies the bounded differences condition with
  parameters $\Delta_i \ge 0$ for $i \in [n]$.  Then for
  $\Delta = \sum_i \Delta_i^{2}$ it holds that
  \begin{align*}
    \Pr[f > \mathbb{E}[f] + t] \le e^{-2 t^2 / \Delta}.
  \end{align*}
\end{theorem}

As before, we are usually interested in showing that a random variable
does not exceed a certain upper bound with high probability.
Analogously to the Chernoff bound in \Cref{col:chernoff}, one can show
that, again, an upper bound on the expected value suffices to show
concentration.

\begin{corollary}
  \label{col:bounded-differences}
  Let $X_1, \dots, X_n$ be independent random variables taking values
  in a set $S$ and let $f \colon S^{n} \rightarrow \mathbb{R}$ be a
  function that satisfies the bounded differences condition with
  parameters $\Delta_i \ge 0$ for $i \in [n]$.  If $g(n)$ is an upper
  bound on $\mathbb{E}[f]$ then for $\Delta = \sum_i \Delta_i^{2}$ and
  $c \ge 1$ it holds that
  \begin{align*}
    \Pr[f > c g(n)] \le e^{-2 ((c - 1)g(n))^2 / \Delta}.
  \end{align*}
\end{corollary}
\begin{proof}
  Let $h(n) \ge 0$ be a function with $f' = f + h(n)$ such that
  $\mathbb{E}[f'] = g(n)$.  Note that $h(n)$ exists since $g(n) \ge
  \mathbb{E}[f]$.  As a consequence, we have $f \le f'$ for all
  outcomes of $X_1, \dots, X_n$ and it holds that 
  \begin{align*}
    \vert f'(\boldsymbol{x}) - f'(\boldsymbol{x}') \vert = \vert f(\boldsymbol{x}) + h(n) - f(\boldsymbol{x}') - h(n) \vert = \vert f(\boldsymbol{x}) - f(\boldsymbol{x}') \vert,
  \end{align*}
  for all $\boldsymbol{x}, \boldsymbol{x}' \in S^{n}$.  Consequently,
  $f'$ satisfies the bounded differences condition with the same
  parameters $\Delta_i$ as $f$.  Since $f \le f'$ it holds that
  \begin{align*}
    \Pr[f > c g(n)] \le \Pr[f' > c g(n)] = \Pr[f' > c \mathbb{E}[f']].
  \end{align*}
  Choosing $t = (c - 1)\mathbb{E}[f']$ allows us to apply
  \Cref{thm:bounded-differences} to conclude that
  \begin{equation*}
    \Pr[f' > c \mathbb{E}[f']] = \Pr[f' > \mathbb{E}[f'] + t] \le e^{-2((c-1)\mathbb{E}[f'])^2 / \Delta} = e^{-2((c-1)g(n))^2 / \Delta}.
  \end{equation*}
\end{proof}

A disadvantage of the method of bounded differences is that one has to
consider the worst possible change in $f$ when changing one variable
and the resulting bound becomes worse the larger this change.  A way
to overcome this issue is to consider the method of \emph{typical}
bounded differences instead.  Intuitively, it allows us to milden the
effect of the change in the worst case, if it is sufficiently
unlikely, and to focus on the typical cases where the change should be
small, instead.  Formally, we say that a function
$f \colon S^{n} \rightarrow \mathbb{R}$ satisfies the \emph{typical
  bounded differences condition} with respect to an event
$A \subseteq S^{n}$ if for all $i \in [n]$ there exist
$\Delta_i^{A} \le \Delta_i$ such that
\begin{align}
  \label{eq:typical-bounded-differences-condition}
  \vert f(\boldsymbol{x}) - f(\boldsymbol{x}') \vert \le
  \begin{cases}
    \Delta_i^{A}, & \text{if}~\boldsymbol{x} \in A,\\
    \Delta_i, & \text{otherwise},
  \end{cases}
\end{align}
for all $\boldsymbol{x}, \boldsymbol{x}' \in S^{n}$ that differ only
in the $i$-th component.

\begin{theorem}[{Method of Typical Bounded Differences,~\cite[Theorem 2\protect\footnotemark]{w-mtbd-16}}]
  \footnotetext{We state a slightly simplified version in order to
    facilitate understandability.  The original theorem allows for the
    random variables $X_1, \dots, X_n$ to take values in different
    sets.}
  \label{thm:typical-bounded-differences}
  Let $X_1, \dots, X_n$ be independent random variables taking values
  in a set $S$ and let $A \subseteq S^{n}$ be an event.  Furthermore,
  let $f \colon S^{n} \rightarrow \mathbb{R}$ be a function that
  satisfies the typical bounded differences condition with respect to
  $A$ and with parameters $\Delta_i^{A} \le \Delta_i$ for $i \in [n]$.
  Then for all $\varepsilon_1, \dots, \varepsilon_n \in (0, 1]$ there
  exists an event~$B$ satisfying $\bar{B} \subseteq A$ and
  $\Pr[B] \le \Pr[\bar{A}] \cdot \sum_{i \in [n]} 1/\varepsilon_i$,
  such that for
  $\Delta = \sum_{i \in [n]} (\Delta_i^{A} + \varepsilon_i (\Delta_i -
  \Delta_i^{A}))^2$ and $t \ge 0$ it holds that
  \begin{align*}
    \Pr[f > \mathbb{E}[f] + t \land \bar{B}] \le e^{-t^2 / (2\Delta)}.
  \end{align*}
\end{theorem}

Intuitively, the choice of the values for $\varepsilon_i$ has two
effects.  On the one hand, choosing $\varepsilon_i$ small allows us to
compensate for a potentially large worst-case change $\Delta_i$.  On
the other hand, this also increases the bound on the probability of
the event $B$ that represents the atypical case.  However, in that
case one can still obtain meaningful bounds if the typical event $A$
occurs with high enough probability.  Again, it is usually sufficient
to show that the function $f$ does not exceed an upper bound on its
expected value with high probability.  The proof of the following
corollary is analogous to the one of \Cref{col:bounded-differences}.

\begin{corollary}[{\cite[Corollary 4.13]{bff-espsf-22}}]
  \label{col:typical-bounded-differences}
  Let $X_1, \dots, X_n$ be independent random variables taking values
  in a set $S$ and let $A \subseteq S^{n}$ be an event.  Furthermore,
  let $f \colon S^{n} \rightarrow \mathbb{R}$ be a function that
  satisfies the typical bounded differences condition with respect to
  $A$ and with parameters $\Delta_i^{A} \le \Delta_i$ for $i \in [n]$
  and let $g(n)$ be an upper bound on $\mathbb{E}[f]$.  Then for all
  $\varepsilon_1, \dots, \varepsilon_n \in (0, 1]$,
  $\Delta = \sum_{i \in [n]} (\Delta_i^{A} + \varepsilon_i (\Delta_i -
  \Delta_i^{A}))^2$, and $c \ge 1$ it holds that
  \begin{align*}
    \Pr[f > c g(n)] \le e^{-((c - 1) g(n))^2 / (2\Delta)} + \Pr[\bar{A}] \sum_{i \in [n]} 1 / \varepsilon_i.
  \end{align*}
\end{corollary}

\paragraph*{Useful Inequalities.}

Finally, computations can often be simplified by making use of the
fact that $1 \pm x$ can be closely approximated by $e^{\pm x}$ for
small $x$.  More precisely, we use the following lemmas, which have
been derived previously using the Taylor
approximation~\cite{k-shrg-16}.

\begin{lemma}[Lemma 2.1, \cite{k-shrg-16}]
  \label{lem:1+x-upper}
  Let $x \in \mathbb{R}$.  Then, $1 + x \le e^{x}$.
\end{lemma}

\begin{lemma}[Corollary of Lemma 2.2, \cite{k-shrg-16}]
  \label{lem:1-x-lower}
  Let $x > 0$ with $x = o(1)$.  Then, $1 - x \ge e^{-(1 + o(1))x}$.
\end{lemma}
\begin{proof}
  First, note that there exists an $\varepsilon = o(1)$ such that $1 -
  x = e^{-\varepsilon}$.  Therefore, it suffices to show that
  $e^{-\varepsilon} \ge e^{-(1 + \varepsilon)x}$.  It is easy to see
  that
  \begin{align*}
    e^{-\varepsilon} &= e^{-\left(1 + \varepsilon - \frac{1 + \varepsilon}{1 + \varepsilon} \right)} = e^{-(1 + \varepsilon)\left(1 - \frac{1}{1 + \varepsilon}\right)}.
  \end{align*}
  By Lemma~\ref{lem:1+x-upper} it holds that $1 + \varepsilon \le
  e^{\varepsilon}$ and, therefore, $1/(1 + \varepsilon) \ge
  e^{-\varepsilon}$.  Since $\varepsilon$ is chosen such that
  $e^{-\varepsilon} = 1 - x$, we can bound $1/(1+\varepsilon) \ge 1 -
  x$ in the above equation and obtain
  \begin{equation*}
    e^{-x} \ge e^{-(1 + \varepsilon)(1 - (1 - x))} = e^{-(1 + \varepsilon)x}.
    \qedhere
  \end{equation*}
\end{proof}

\begin{lemma}[Lemma 2.3, \cite{k-shrg-16}]
  \label{lem:1+x-frac-bound}
  Let $x \in \mathbb{R}$ with $x = \pm o(1)$. Then, $1/(1 + x) = 1 -
  \Theta(x)$.
\end{lemma}

\section{An Improved Greedy Algorithm}
\label{sec:algorithm}

Previous insights about solving the vertex cover problem on hyperbolic
random graphs are based on the fact that the \emph{dominance reduction
  rule} reduces the graph to a remainder of simple
structure~\cite{bffk-svcpt-20}.  This rule states that a vertex $u$
can be safely added to the vertex cover (and, thus, be removed from
the graph) if it \emph{dominates} at least one other vertex, i.e., if
there exists a neighbor $v \in N(u)$ such that all neighbors of $v$
are also neighbors of $u$.

On hyperbolic random graphs, vertices near the center of the disk
dominate with high probability~\cite[Lemma 5]{bffk-svcpt-20}.
Therefore, it is not surprising that the standard greedy algorithm
that computes a vertex cover by repeatedly taking the vertex with the
largest degree achieves good approximation rates on such networks:
Since high degree vertices are near the disk center, the algorithm
essentially favors vertices that are likely to dominate and can be
safely added to the vertex cover anyway.

On the other hand, after (safely) removing high-degree vertices, the
remaining vertices all have similar (small) degree, meaning the
standard greedy algorithm basically picks the vertices at random.
Thus, in order to improve the approximation performance of the
algorithm, one has to improve on the parts of the graph that contain
the low-degree vertices.  Based on this insight, we derive a new
greedy algorithm that achieves close to optimal approximation rates
efficiently.  More formally, we prove the following main theorem.

\wormhole{thm-vertex-cover-efficient-approximation}
\begin{theorem}
  \label{thm:vertex-cover-cover-efficient-approximation}
  Let $G$ be a hyperbolic random graph on $n$ vertices.  Given the
  radii of the vertices, an approximate vertex cover of $G$ can be
  computed in time $\mathcal{O}(m \log(n))$, such that the
  approximation ratio is $(1 + o(1))$ asymptotically almost surely.
\end{theorem}

Consider the following greedy algorithm that computes an approximation
of a minimum vertex cover on hyperbolic random graphs.  We iterate the
vertices in order of increasing radius.  Each encountered vertex $v$
is added to the cover and removed from the graph.  After each step, we
then identify the connected components of size at most
$\tau \log\log(n)$ in the remainder of the graph, solve them
optimally, and remove them from the graph as well.  The constant
$\tau > 0$ can be used to adjust the trade-off between quality and
running time: With increasing $\tau$ the parts of the graph that are
solved exactly increase as well, but so does the running time.

This algorithm determines the order in which the vertices are
processed based on their radii, which are not known for real-world
networks.  However, in hyperbolic random graphs, there is a strong
correlation between the radius of a vertex and its
degree~\cite{gpp-rhg-12}.  Therefore, we can mimic the considered
greedy strategy by removing vertices with decreasing degree instead.
Then, the above algorithm represents an adaptation of the standard
greedy algorithm: Instead of greedily adding vertices with decreasing
degree until all remaining vertices are isolated, we increase the
quality of the approximation by solving small components exactly.

\section{Approximation Performance}
\label{sec:approximation}

To analyze the performance of the above algorithm, we utilize
structural properties of hyperbolic random graphs.  While the
power-law degree distribution and high clustering are modelled
explicitly using the underlying geometry, other properties of the
model, like the logarithmic diameter, emerge as a natural consequence
of the first two.  Our analysis is based on another emerging property:
Hyperbolic random graphs decompose into small components when removing
high-degree vertices.

More formally, we proceed as follows.  We compute the size of the
vertex cover obtained using the above algorithm, by partitioning the
vertices of the graph into two sets: $V_{\text{Greedy}}$ and
$V_{\text{Exact}}$, denoting the vertices that were added greedily and
the ones contained in small separated components that were solved
exactly, respectively (see \Cref{fig:graph}).  Clearly, we obtain a
valid vertex cover for the whole graph, if we take all vertices in
$V_{\text{Greedy}}$ together with a vertex cover $C_{\text{Exact}}$ of
$G[V_{\text{Exact}}]$.  Then, the approximation ratio is given by the
quotient
$\delta = (\vert V_{\text{Greedy}} \vert + \vert C_{\text{Exact}}
\vert)/ \vert C_{\text{OPT}}\vert,$
where $C_{\text{OPT}}$ denotes an optimal solution.  Since all
components in $G[V_{\text{Exact}}]$ are solved optimally and since any
minimum vertex cover for the whole graph induces a vertex cover on
$G[V']$ for any vertex subset $V' \subseteq V$, it holds that
$\vert C_{\text{Exact}} \vert \le \vert C_{\text{OPT}} \vert$.
Consequently, it suffices to show that
$\vert V_{\text{Greedy}} \vert \in o(\vert C_{\text{OPT}}\vert )$ in
order to obtain the claimed approximation factor of $1 + o(1)$.

\begin{figure}[t]
  \centering
  \includegraphics[width=0.8\linewidth]{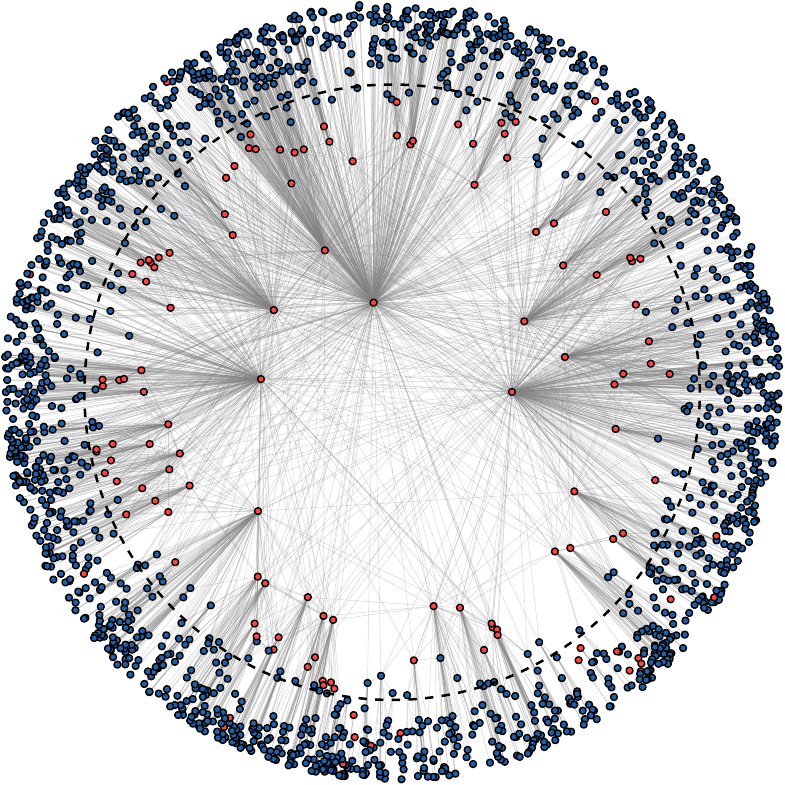}
  \caption{A hyperbolic random graph with $1942$ vertices, average
    degree $7.7$, and power-law exponent $2.6$.  The vertex sets
    $V_{\text{Greedy}}$ and $V_{\text{Exact}}$ are shown in red and
    blue, respectively.  The dashed line shows a possible threshold
    radius $\rho$.}
  \label{fig:graph}
\end{figure}

To bound the size of $V_{\text{Greedy}}$, we identify a time during
the execution of the algorithm at which only few vertices were added
greedily, yet, the majority of the vertices were contained in small
separated components (and were, therefore, part of
$V_{\text{Exact}}$), and only few vertices remain to be added
greedily.  Since the algorithm processes the vertices by increasing
radius, this point in time can be translated to a threshold radius
$\rho$ in the hyperbolic disk (see \Cref{fig:graph}).  Therefore, we
divide the hyperbolic disk into two regions: an \emph{inner disk} and
an \emph{outer band}, containing vertices with radii below and above
$\rho$, respectively.  The threshold~$\rho$ is chosen such that a
hyperbolic random graph decomposes into small components after
removing the inner disk.  When adding the first vertex from the outer
band, greedily, we can assume that the inner disk is empty (since
vertices of smaller radii were chosen before or removed as part of a
small component).  At this point, the majority of the vertices in the
outer band were contained in small components, which have been solved
exactly.  In our analysis, we now overestimate the size of
$V_{\text{Greedy}}$ by assuming that all remaining vertices are also
added to the cover greedily.  Therefore, we obtain a valid upper bound
on $\vert V_{\text{Greedy}} \vert$, by counting the total number of
vertices in the inner disk and adding the number of vertices in the
outer band that are contained in components that are not solved
exactly, i.e., components whose size exceeds $\tau \log\log(n)$.  In
the following, we show that both numbers are sublinear in $n$ with
high probability.  Together with the fact that an optimal vertex cover
on hyperbolic random graphs, asymptotically almost surely, contains
$\Omega(n)$ vertices~\cite{cfr-ggdsfn-16}, this implies
$\vert V_{\text{Greedy}} \vert \in o(\vert C_{\text{OPT}} \vert)$.

\begin{figure}
  \centering
  \includegraphics{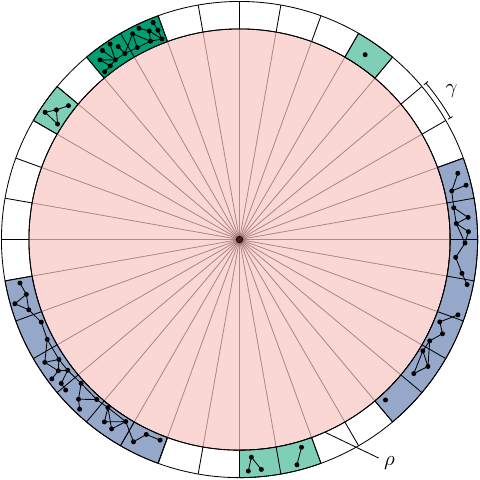}
  \caption{The disk is divided into the inner disk (red) and the outer
    band.  It is additionally divided into sectors of equal width
    $\gamma$.  Consecutive non-empty sectors form a run.  Wide runs
    (blue) consist of many sectors.  Each blue sector is a widening
    sector.  Narrow runs (green) consist of few sectors.  Small narrow
    runs contain only few vertices (light green), while large narrow
    runs contain many vertices (dark green).}
  \label{fig:sectors}
\end{figure}

The main contribution of our analysis is the identification of small
components in the outer band, which is done by discretizing it into
sectors, such that an edge cannot extend beyond an empty sector (see
\Cref{fig:sectors}).  The foundation of this analysis is the delicate
interplay between the angular width $\gamma$ of these sectors and the
threshold $\rho$ that defines the outer band.  Recall that $\rho$ is
used to represent the time in the execution of the algorithm at which
the graph has been decomposed into small components.  For our analysis
we assume that all vertices seen before this point (all vertices in
the inner disk; red \Cref{fig:sectors}) were added greedily.
Therefore, if we choose $\rho$ too large, we overestimate the actual
number of greedily added vertices by too much.  As a consequence, we
want to choose $\rho$ as small as possible.  However, this conflicts
our intentions for the choice of $\gamma$ and its impact on $\rho$.
Recall that the maximum angular distance between two vertices such
that they are adjacent increases with decreasing radii
(\Cref{eq:maximum-angular-distance}).  Thus, in order to avoid edges
that extend beyond an angular width of~$\gamma$, we need to ensure
that the radii of the vertices in the outer band are sufficiently
large.  That is, decreasing $\gamma$ requires increasing $\rho$.
However, we want to make $\gamma$ as small as possible, in order to
get a finer granularity in the discretization and, with that, a more
accurate analysis of the component structure in the outer band.
Therefore, $\gamma$ and $\rho$ need to be chosen such that the inner
disk does not become too large, while ensuring that the discretization
is granular enough to accurately detect components whose size depends
on $\tau$ and $n$.  To this end, we adjust the angular width of the
sectors using a function $\gamma(n, \tau)$, which is defined as
\begin{align*}
  \gamma(n, \tau) = \log \left(\frac{\tau \log^{(2)}(n)}{2 \log^{(3)}(n)^2} \right),
\end{align*}
where $\log^{(i)}(n)$ denotes iteratively applying the $\log$-function
$i$ times on $n$ (e.g., $\log^{(2)}(n) = \log\log(n)$), and set
\begin{align*}
  \rho = R - \log(\pi/2 \cdot e^{C/2} \gamma(n, \tau)),
\end{align*}
where $R = 2 \log(n) + C$ is the radius of the hyperbolic disk.

In the following, we first show that the number of vertices in the
inner disk is sublinear with high probability, before analyzing the
component structure in the outer band.  To this end, we make use of
the discretization of the disk into sectors, by distinguishing between
different kinds of \emph{runs} (sequences of non-empty sectors), see
\Cref{fig:sectors}.  In particular, we bound the number of \emph{wide}
runs (consisting of many sectors) and the number of vertices in them.
Then we bound the number of vertices in \emph{large narrow} runs
(consisting of few sectors but containing many vertices).  The
remaining \emph{small narrow} runs represent small components that are
solved exactly.

The analysis mainly involves working with the random variables that
denote the numbers of vertices in the above mentioned areas of the
disk.  Throughout the paper, we usually start with computing their
expected values.  Afterwards, we obtain tight concentration bounds
using the previously mentioned Chernoff bound or, when the considered
random variables are more involved, the method of (typical) bounded
differences.

\subsection{The Inner Disk}

The inner disk $\mathcal{I}$ contains all vertices whose radius is
below the threshold $\rho$.  The number of them that are added to the
cover greedily is bounded by the number of all vertices in
$\mathcal{I}$.

\begin{lemma}
  \label{lem:inner-disk-bad-nodes-whp}
  Let $G$ be a hyperbolic random graph on $n$ vertices with power-law
  exponent $\beta = 2\alpha + 1$.  With high probability, the number
  of vertices in $\mathcal{I}$ is in
  $\mathcal{O}(n \cdot \gamma(n, \tau)^{-\alpha})$.
\end{lemma}
\begin{proof}
  We start by computing the expected number of vertices in
  $\mathcal{I}$ and show concentration afterwards.  To this end, we
  first compute the measure $\mu(\mathcal{I})$.  The measure of a disk
  of radius~$r$ that is centered at the origin is given by $e^{-\alpha
    (R - r)}(1 + o(1))$ \cite[Lemma 3.2]{gpp-rhg-12}.  Consequently,
  the expected number of vertices in $\mathcal{I}$ is
  \begin{align*}
    \mathbb{E}[\vert V(\mathcal{I}) \vert] &= n \mu(\mathcal{I}) \\
                                           &= \mathcal{O}(n e^{-\alpha (R - \rho)}) \\
                                           &= \mathcal{O}(n e^{- \alpha \log(\pi/2 \SpaciousCdot e^{C/2} \gamma(n, \tau))}) \\
                                           &= \mathcal{O}\left(n \cdot \gamma(n, \tau)^{- \alpha} \right).
  \end{align*}
  Since $\gamma(n, \tau) = \mathcal{O}(\log^{(3)}(n))$, this bound on
  $\mathbb{E}[\vert V(\mathcal{I}) \vert]$ is $\omega(\log(n))$, and
  we can apply the Chernoff bound in \Cref{col:chernoff} to conclude
  that
  $\vert V(\mathcal{I}) \vert = \mathcal{O}\left(n \cdot \gamma(n,
    \tau)^{- \alpha} \right)$ holds with probability
  $1 - \mathcal{O}(n^{-c})$ for any $c > 0$.
\end{proof}

Since $\gamma(n, \tau) = \omega(1)$,
\Cref{lem:inner-disk-bad-nodes-whp} shows that, with high probability,
the number of vertices that are greedily added to the vertex cover in
the inner disk is sublinear.  Once the inner disk has been processed
and removed, the graph has been decomposed into small components and
the ones of size at most $\tau \log\log(n)$ have already been solved
exactly.  The remaining vertices that are now added greedily belong to
large components in the outer band.

\subsection{The Outer Band}

To identify the vertices in the outer band that are contained in
components whose size exceeds $\tau \log\log(n)$, we divide it into
sectors of angular width
$\gamma = \theta(\rho, \rho) = \pi \cdot \gamma(n, \tau)/n \cdot (1 +
o(1))$, where $\theta(\rho, \rho)$ denotes the maximum angular
distance between two vertices with radii $\rho$ to be adjacent (see
\Cref{eq:maximum-angular-distance}).  This division is depicted in
\Cref{fig:sectors}.  The choice of~$\gamma$ (combined with the choice
of $\rho$) has the effect that an edge between two vertices in the
outer band cannot extend beyond an empty sector, i.e., a sector that
does not contain any vertices, allowing us to use empty sectors as
delimiters between components.  To this end, we introduce the notion
of \emph{runs}, which are maximal sequences of non-empty sectors (see
\Cref{fig:sectors}).  While a run can contain multiple components, the
number of vertices in it denotes an upper bound on the combined sizes
of the components that it contains.

To show that there are only few vertices in components whose size
exceeds $\tau \log\log(n)$, we bound the number of vertices in runs
that contain more than $\tau \log\log(n)$ vertices.  For a given run
this can happen for two reasons.  First, it may contain many vertices
if its angular interval is too large, i.e., it consists of too many
sectors.  This is unlikely, since the sectors are chosen sufficiently
small, such that the probability for a given one to be empty is high.
Second, while the angular width of the run is not too large, it
contains too many vertices for its size.  However, the vertices of the
graph are distributed uniformly at random in the disk, making it
unlikely that too many vertices are sampled into such a small area.
To formalize this, we introduce a threshold $w$ and distinguish
between two types of runs: A \emph{wide} run contains more than $w$
sectors, while a \emph{narrow} run contains at most $w$ sectors.  The
threshold~$w$ is chosen such that the probabilities for a run to be
wide and for a narrow run to contain more than $\tau \log\log(n)$
vertices are small.  To this end, we set
$w = e^{\gamma(n, \tau)} \cdot \log^{(3)}(n)$.

In the following, we first bound the number of vertices in wide runs.
Afterwards, we consider narrow runs that contain more than
$\tau \log\log(n)$ vertices.  Together, this gives an upper bound on
the number of vertices that are added greedily in the outer band.

\subsubsection{Wide Runs}

We refer to a sector that contributes to a wide run as a
\emph{widening sector}.  In the following, we bound the number of
vertices in all wide runs in three steps.  First, we determine the
expected number of all widening sectors.  Second, based on the
expected value, we show that the number of widening sectors is small,
with high probability.  Finally, we make use of the fact that the area
of the disk covered by widening sectors is small, to show that the
number of vertices sampled into the corresponding area is sublinear,
with high probability.

\subparagraph{Expected Number of Widening Sectors.}

Let $n'$ denote the total number of sectors and let $\mathcal{S}_1,
\dots, \mathcal{S}_{n'}$ be the corresponding sequence.  For each
sector $\mathcal{S}_k$, we define the random variable $S_k$ indicating
whether $\mathcal{S}_k$ contains any vertices, i.e., $S_k = 0$ if
$\mathcal{S}_k$ is empty and $S_k = 1$ otherwise.  The sectors in the
disk are then represented by a circular sequence of indicator random
variables $S_1, \dots, S_{n'}$, and we are interested in the random
variable $W$ that denotes the sum of all runs of $1$s that are longer
than $w$.  In order to compute $\mathbb{E}[W]$, we first compute the
total number of sectors, as well as the probability for a sector to be
empty or non-empty.

\begin{lemma}
  \label{lem:num-sectors}
  Let $G$ be a hyperbolic random graph on $n$ vertices. Then, the
  number of sectors of width $\gamma = \theta(\rho, \rho)$ is $n' = 2n
  / \gamma(n, \tau) \cdot (1 \pm o(1))$.
\end{lemma}
\begin{proof}
  Since all sectors have equal angular width
  $\gamma = \theta(\rho, \rho)$, we can use
  \Cref{eq:maximum-angular-distance} to compute the total number of
  sectors as
  \begin{align*}
    n' &= 2\pi/\theta(\rho, \rho) \\
       &= \pi e^{-R/2 + \rho} (1 \pm \Theta(e^{R - 2\rho}))^{-1}.
  \end{align*}
  By substituting $\rho = R - \log(\pi/2 \cdot e^{C/2} \gamma(n,
  \tau))$ and $R = 2\log(n) + C$, we obtain
  \begin{align*}
    n' &= \frac{\pi e^{R/2}}{\pi/2 \cdot e^{C/2} \gamma(n, \tau)} (1 \pm \Theta (e^{-R} \gamma(n, \tau)^2 ))^{-1} \\
       &= 2n/\gamma(n, \tau) \cdot (1 \pm \Theta ((\gamma(n, \tau)/n)^2))^{-1}.
  \end{align*}
  It remains to simplify the error term.  Note that
  $\gamma(n, \tau) = \mathcal{O}(\log^{(3)}(n))$.  Consequently, the
  error term is equivalent to $(1 \pm o(1))^{-1}$.  Finally, it holds
  that $1/(1 + x) = 1 - \Theta(x)$ for $x = \pm o(1)$.
\end{proof}

\begin{lemma}
  \label{lem:prob-sector-non-empty}
  Let $G$ be a hyperbolic random graph on $n$ vertices and let
  $\mathcal{S}$ be a sector of angular width $\gamma = \theta(\rho,
  \rho)$.  For sufficiently large $n$, the probability that
  $\mathcal{S}$ contains at least one vertex is bounded by
  \begin{align*}
    1 - e^{-\gamma(n, \tau) / 4} \le \Pr[V(\mathcal{S}) \neq \emptyset] \le e^{-\left(e^{-\gamma(n, \tau)} \right)}.
  \end{align*}
\end{lemma}
\begin{proof}
  To compute the probability that $\mathcal{S}$ contains at least one
  vertex, we first compute the probability for a given vertex to fall
  into $\mathcal{S}$, which is given by the measure
  $\mu(\mathcal{S})$.  Since the angular coordinates of the vertices
  are distributed uniformly at random and since the disk is divided
  into $n'$ sectors of equal width, the measure of a single sector
  $\mathcal{S}$ can be obtained as $\mu(\mathcal{S}) = 1/n'$.  The
  total number of sectors $n'$ is given by \Cref{lem:num-sectors} and
  we can derive
  \begin{align*}
    \mu(\mathcal{S}) &= \frac{\gamma(n, \tau)}{2n} (1 \pm o(1))^{-1} = \frac{\gamma(n, \tau)}{2n} (1 \pm o(1)),
  \end{align*}
  where the second equality is obtained by applying
  $1/(1 + x) = 1 - \Theta(x)$ for $x = \pm o(1)$.

  Given $\mu(\mathcal{S})$, we first compute the lower bound on the
  probability that $\mathcal{S}$ contains at least one vertex.  Note
  that
  $\Pr[V(\mathcal{S}) \neq \emptyset] = 1 - \Pr[V(\mathcal{S}) =
  \emptyset]$.  Therefore, it suffices to show that
  $\Pr[V(\mathcal{S}) = \emptyset] \le e^{-\gamma(n, \tau)/4}$.  The
  probability that $\mathcal{S}$ is empty is
  $(1 - \mu(\mathcal{S}))^n$.  Now recall that $1 - x \le e^{-x}$ for
  all $x \in \mathbb{R}$.  Consequently, we have
  \begin{align*}
    \Pr[V(\mathcal{S}) = \emptyset] \le e^{-n\mu(\mathcal{S})} \le e^{-\gamma(n, \tau)/2 \SpaciousCdot (1 - o(1))}
  \end{align*}
  and for large enough $n$ it holds that $1 - o(1) \ge 1/2$.

  It remains to compute the upper bound.  Again, since
  $\Pr[V(\mathcal{S}) \neq \emptyset] = 1 - \Pr[V(\mathcal{S}) =
  \emptyset]$ and since $\Pr[V(\mathcal{S}) = \emptyset] = (1 -
  \mu(\mathcal{S}))^n$, we can compute the probability that
  $\mathcal{S}$ contains at least one vertex as
  \begin{align*}
    \Pr[V(\mathcal{S}) \neq \emptyset] = 1 - (1 - \mu(\mathcal{S}))^n.
  \end{align*}
  Note that $\mu(\mathcal{S}) \in o(1)$.  Therefore, we can bound
  $1 - x \ge e^{-x(1 + o(1))}$ for $x \in o(1)$, and obtain the
  following bound on $\Pr[V(\mathcal{S}) \neq \emptyset]$
  \begin{align*}
    \Pr[V(\mathcal{S}) \neq \emptyset] &= 1 - (1 - \mu(\mathcal{S}))^n \\
                                       &\le 1 - e^{-n \mu(\mathcal{S}) (1 + o(1))} \\
                                       &\le 1 - e^{-\gamma(n, \tau)/2 \SpaciousCdot (1 + o(1))}.
  \end{align*}
  For large enough $n$, we have $(1 + o(1)) \le 2$.  Therefore,
  \begin{align*}
    \Pr[V(\mathcal{S}) \neq \emptyset] &\le 1 - e^{-\gamma(n, \tau)}
  \end{align*}
  holds for sufficiently large $n$.  Finally, $1 - x \le e^{-x}$ is
  valid for all $x \in \mathbb{R}$ and we obtain the claimed bound.
\end{proof}

We are now ready to bound the expected number of widening sectors,
i.e., sectors that are part of wide runs.  To this end, we aim to
apply the following lemma.

\begin{lemma}[{\cite[Proposition 4.3\protect\footnotemark]{mpp-s-07}}]
  \footnotetext{The original statement has been adapted to fit our
    notation.  We use $n', w$, and $W$ to denote the total number of
    random variables, the threshold for long runs, and the sum of
    their lengths, respectively.  They were previously denoted by $n,
    k$, and $S$, respectively.  In the original statement $s = 0$
    indicates that the variables are distributed independently and
    identically, and $c$ indicates that the sequence is circular.}
  \label{lem:bernoulli-sequence}
  Let $S_1, \dots, S_{n'}$ denote a circular sequence of independent
  indicator random variables, such that $\Pr[S_k = 1] = p$ and
  $\Pr[S_k = 0] = 1 - p = q$, for all $k \in [n']$.  Furthermore, let
  $W$ denote the sum of the lengths of all success runs of length at
  least $w \le n'$.  Then, $\mathbb{E}[W] = n' p^{w} (wq + p)$.
\end{lemma}

We note that the indicator random variables $S_1, \dots, S_n'$ are
\emph{not} independent on hyperbolic random graphs.  To overcome this
issue, we compute the expected value of $W$ on hyperbolic random
graphs with $n$ vertices in expectation (see \Cref{sec:preliminaries})
and subsequently derive a probabilistic bound on $W$ for hyperbolic
random graphs.

\begin{lemma}
  \label{lem:expected-bad-sectors}
  Let $G$ be a hyperbolic random graph with $n$ vertices in
  expectation and let $W$ denote the number of widening sectors.
  Then,
  \begin{align*}
    \mathbb{E}[W] \le \frac{2^{1/4} \cdot \tau^{3/4} \cdot n}{\gamma(n, \tau) \cdot \log^{(2)}(n)^{1/4} \cdot \log^{(3)}(n)^{1/2}} (1 \pm o(1)).
  \end{align*}
\end{lemma}
\begin{proof}
  A widening sector is part of a run of more than
  $w = e^{\gamma(n, \tau)} \cdot \log^{(3)}(n)$ consecutive non-empty
  sectors.  To compute the expected number of widening sectors, we
  apply \Cref{lem:bernoulli-sequence}.  To this end, we use
  \Cref{lem:num-sectors} to bound the total number of sectors $n'$ and
  bound the probability $p = \Pr[S_k = 1]$ (i.e., the probability that
  sector $\mathcal{S}_k$ is not empty) as
  $p \le \exp(-(e^{-\gamma(n, \tau)}))$, as well as the complementary
  probability $q = 1 - p \le e^{-\gamma(n, \tau)/4}$, using
  \Cref{lem:prob-sector-non-empty}.  We obtain
  \begin{align*}
    \mathbb{E}[W] &= n' p^{(w + 1)} ((w + 1)q + p) \\
                  &\le \frac{2n}{\gamma(n, \tau)} (1 \pm o(1)) \cdot e^{-\left( (w + 1) e^{-\gamma(n, \tau)} \right)} \cdot \left( (w + 1)e^{-\frac{\gamma(n, \tau)}{4}} + 1 \right) \\
                  &\le \frac{2n}{\gamma(n, \tau)} e^{\left( -e^{\gamma(n, \tau)} \log^{(3)}(n) e^{-\gamma(n, \tau)} \right)} \\
                  &\hphantom{\le \frac{2n}{\gamma(n, \tau)}~} \cdot \left( (e^{\gamma(n, \tau)} \log^{(3)}(n) + 1) e^{-\frac{\gamma(n, \tau)}{4}} + 1 \right) (1 \pm o(1)).
  \end{align*}
  Now the first exponential simplifies to
  $\exp(-\log^{(3)}(n)) = \log^{(2)}(n)^{-1}$, since
  the~$\gamma(n, \tau)$ terms cancel.  Factoring out
  $\exp(3/4 \cdot \gamma(n, \tau)) \log^{(3)}(n)$ in the third term
  then yields
  \begin{align*}
    \mathbb{E}[W] &\le \frac{2n e^{3/4 \SpaciousCdot \gamma(n, \tau)} \log^{(3)}(n)}{\gamma(n, \tau) \cdot \log^{(2)}(n)} \\
                  &\qquad \cdot \left( 1 + \frac{1}{e^{\gamma(n, \tau)} \log^{(3)}(n)} + \frac{1}{e^{3/4 \SpaciousCdot \gamma(n, \tau)} \log^{(3)}(n)} \right) (1 \pm o(1)).
  \end{align*}
  Since $\gamma(n, \tau) = \omega(1)$, the first error term can be
  simplified as $(1 + o(1))$.  Additionally, we can substitute
  $\gamma(n, \tau) = \log(\tau \log^{(2)}(n)/(2 \log^{(3)}(n)^2))$ to
  obtain
  \begin{align*}
    \mathbb{E}[W] &\le 2^{1/4} \frac{ \tau^{3/4} \cdot n \cdot \log^{(3)}(n)}{\gamma(n, \tau) \cdot \log^{(2)}(n)} \cdot \frac{\log^{(2)}(n)^{3/4}}{\log^{(3)}(n)^{3/2}} \cdot (1 \pm o(1)).
  \end{align*}
  Further simplification then yields the claim.
\end{proof}

\subparagraph{Concentration Bound on the Number of Widening Sectors.}

\Cref{lem:expected-bad-sectors} bounds the expected number of widening
sectors and it remains to show that this bound holds with high
probability.  To this end, we first determine under which conditions
the sum of long success runs in a circular sequence of indicator
random variables can be bounded with high probability in general.
Afterwards, we show that these conditions are met for our application.

\begin{lemma}
  \label{lem:whp-circular-success-runs}
  Let $S_1, \dots, S_{n'}$ denote a circular sequence of independent
  indicator random variables and let $W$ denote the sum of the lengths
  of all success runs of length at least $1 \le w \le n'$.  If $g(n')
  = \omega(w \sqrt{n' \log(n')})$ is an upper bound on
  $\mathbb{E}[W]$, then $W = \mathcal{O}(g(n'))$ holds with
  probability $1 - \mathcal{O}((n')^{-c})$ for any constant $c$.
\end{lemma}
\begin{proof}
  In order to show that $W$ does not exceed $g(n')$ by more than a
  constant factor with high probability, we aim to apply a method of
  bounded differences (\Cref{col:bounded-differences}).  To this end,
  we consider $W$ as a function of $n'$ independent random variables
  $S_1, \dots, S_{n'}$ and determine the parameters $\Delta_i$ with
  which $W$ satisfies the bounded differences condition (see
  \Cref{eq:bounded-differences-condition}).  That is, for each
  $i \in [n']$ we need to bound the change in the sum of the lengths
  of all success runs of length at least $w$, obtained by changing the
  value of $S_i$ from $0$ to $1$ or vice versa.

  \begin{figure}
    \centering
    \includegraphics{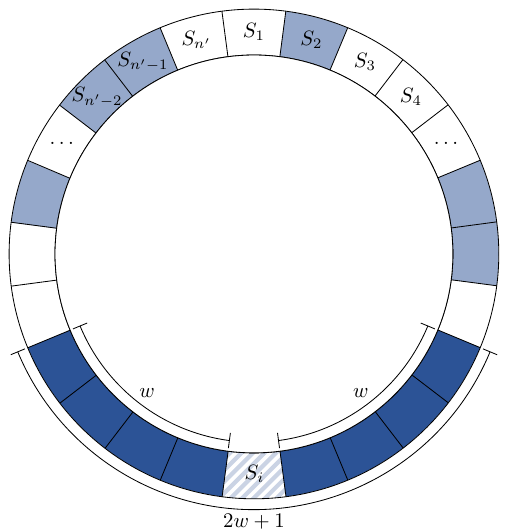}
    \caption{A circular sequence of random variables
      $S_1, \dots, S_{n'}$ that can either be $0$ (white) or~$1$
      (blue).  Dark blue runs are as large as possible without being
      wide.  Depending on the value of $S_i$, the two runs of length
      $w$ are merged into one run of length $2w + 1$.}
    \label{fig:circular-sequence}
  \end{figure}

  The largest impact on $W$ is obtained when changing the value of
  $S_i$ from $0$ to~$1$ merges two runs of size $w$, i.e., runs that
  are as large as possible but not \emph{wide}, as shown in
  \Cref{fig:circular-sequence}.  In this case both runs did not
  contribute anything to $W$ before the change, while the merged run
  now contributes $2w + 1$.  Then, we can bound the change in $W$ as
  $\Delta_i = 2w + 1$.  Note that the other case in which the value of
  $S_i$ is changed from $1$ to $0$ can be viewed as the inversion of
  the change in the first case.  That is, instead of merging two runs,
  changing~$S_i$ splits a single run into two.  Consequently, the
  corresponding bound on the change of~$W$ is the same, except that
  $W$ is decreasing instead of increasing.

  It follows that $W$ satisfies the bounded differences condition for
  $\Delta_i = 2w + 1$ for all $i \in \{1, \dots, n'\}$.  We can now
  apply \Cref{col:bounded-differences} to bound the probability that
  $W$ exceeds an upper bound $g(n')$ on its expected value by more
  than a constant factor as
  \begin{align*}
    \Pr[W > c_1 g(n')] \le e^{-2((c_1 - 1)g(n'))^2 / \Delta},
  \end{align*}
  where $\Delta = \sum_i \Delta_i^2$ and $c_1 \ge 1$.  Since
  $\Delta_i = 2w + 1$ for all $i \in [n']$, we have
  $\Delta = n'(2w + 1)^2$. Thus,
  \begin{align*}
    \Pr[W > c_1 g(n')] \le e^{-\frac{2((c_1 - 1)g(n'))^2}{n' (2w + 1)^2}} \le e^{-\frac{2(c_1 - 1)^2}{n'} \SpaciousCdot \left( \frac{g(n')}{3w} \right)^2},
  \end{align*}
  where the second inequality is valid since $w$ is assumed to be at
  least $1$.  Moreover, we can apply
  $g(n') = \omega(w \sqrt{n' \log(n')})$ (a precondition of this
  lemma), which yields
  \begin{align*}
    \Pr[W > c_1 g(n')] = (n')^{- \omega(1)}.
  \end{align*}
  Therefore, it holds that $\Pr[W \in \mathcal{O}(g(n'))] = 1 -
  (n')^{-\omega(1)} = 1 - \mathcal{O}((n')^{-c})$ for any
  constant~$c$.
\end{proof}

\begin{lemma}
  \label{lem:whp-bad-sectors}
  Let $G$ be a hyperbolic random graph on $n$ vertices.  Then, with
  probability $1 - \mathcal{O}(n^{-c})$ for any constant $c > 0$, the
  number of widening sectors is
  \begin{align*}
    W = \mathcal{O} \left( \frac{\tau^{3/4} \cdot n}{\gamma(n, \tau) \cdot \log^{(2)}(n)^{1/4} \cdot \log^{(3)}(n)^{1/2}} \right).
  \end{align*}
\end{lemma}
\begin{proof}
  In the following, we show that the claimed bound holds with
  probability $1 - \mathcal{O}(n^{-c_1})$ for any constant $c_1 > 0$
  on hyperbolic random graphs with $n$ vertices in expectation.  By
  \Cref{lem:hrg-in-expectation} the same bound then holds with
  probability $1 - \mathcal{O}(n^{-c_1 + 1/2})$ on hyperbolic random
  graphs.  Choosing $c = c_1 - 1/2$ then yields the claim.

  Recall that we represent the sectors using a circular sequence of
  independent indicator random variables $S_1, \dots, S_{n'}$ and that
  $W$ denotes the sum of the lengths of all success runs spanning more
  than $w$ sectors, i.e., the sum of all widening sectors.  By
  \Cref{lem:expected-bad-sectors} we obtain a valid upper bound on
  $\mathbb{E}[W]$ by choosing
  \begin{align*}
    g(n') = h(n) = \frac{2^{1/4} \cdot \tau^{3/4} \cdot n}{\gamma(n, \tau) \cdot \log^{(2)}(n)^{1/4} \cdot \log^{(3)}(n)^{1/2}} (1 \pm o(1))
  \end{align*}
  and it remains to show that this bound holds with sufficiently high
  probability.  To this end, we aim to apply
  \Cref{lem:whp-circular-success-runs}, which states that
  $W = \mathcal{O}(g(n'))$ holds with probability
  $1 - \mathcal{O}((n')^{-c_2})$ for any constant $c_2$, if
  $g(n') = \omega(w \sqrt{n' \log(n')})$.  In the following, we first
  show that $h(n)$ fulfills this criterion\footnote{Since we are
    interested in runs of \emph{more than} $w$ sectors, we need to
    show $g(n') = \omega((w + 1) \sqrt{n' \log(n')})$.  However, it is
    easy to see that this is implied by showing
    $g(n') = \omega(w \sqrt{n' \log(n')})$.}, before arguing that we
  can choose $c_2$ such that
  $1 - \mathcal{O}((n')^{-c_2}) = 1 - \mathcal{O}(n^{-c_1})$ for any
  constant $c_1$.

  Since $\tau$ is constant and
  $n' = 2n/\gamma(n, \tau) \cdot (1 \pm o(1))$ by
  \Cref{lem:num-sectors}, we can bound $h(n)$ by
  \begin{align*}
    h(n) &= \Theta \left( \frac{n'}{\log^{(2)}(n)^{1/4} \log^{(3)}(n)^{1/2}} \right) \\
         &= \Theta \left( \frac{\log^{(2)}(n) \cdot n'}{\log^{(2)}(n)^{5/4} \log^{(3)}(n)^{1/2}} \right) \\
         &= \omega \left( \frac{\log^{(2)}(n)}{\log^{(3)}(n)} \cdot \frac{n'}{\log^{(2)}(n)^{5/4}} \right),
  \end{align*}
  where the last bound is obtained by applying
  $\log^{(3)}(n)^{1/2} = \omega(1)$.  Recall that $w$ was chosen as
  $w = e^{\gamma(n, \tau)} \log^{(3)}(n)$.  Furthermore, we have
  $\gamma(n, \tau) = \log(\tau \log^{(2)}(n) / (2 \log^{(3)}(n)^2))$.
  Thus, it holds that $w = \Theta(\log^{(2)}(n)/(\log^{(3)}(n)))$,
  allowing us to further bound $h(n)$ by
  \begin{align*}
    h(n) &= \omega \left( w \frac{n'}{\log^{(2)}(n)^{5/4}} \right) \\
         &= \omega \left( w \sqrt{n' \log(n') \cdot \frac{n'}{\log(n') \log^{(2)}(n)^{5/2}}} \right)
  \end{align*}
  and it remains to show that the last factor in the root is in
  $\omega(1)$.  Note that $n' = \Omega(n / \log^{(3)}(n))$ and $n' =
  \mathcal{O}(n)$.  Consequently, it holds that
  \begin{align*}
    \frac{n'}{\log(n') \log^{(2)}(n)^{5/2}} &= \Omega \left( \frac{n}{\log(n) \cdot \log^{(2)}(n)^{5/2} \cdot \log^{(3)}(n)} \right) = \omega \left(\frac{n}{\log(n)^3} \right) = \omega(1).
  \end{align*}
  As stated above, this shows that $W = \mathcal{O}(h(n))$ holds with
  probability $1 - \mathcal{O}((n')^{-c_2})$ for any constant $c_2$.
  Again, since $n' = \Omega(n / \log^{(3)}(n))$, we have
  $n' = \Omega(n^{1/2})$.  Therefore, we can conclude that
  $W = \mathcal{O}(h(n))$ holds with probability
  $1 - \mathcal{O}(n^{-c_2/2})$.  Choosing $c_2 = 2 c_1$ then yields
  the claim.
\end{proof}

\subparagraph{Number of Vertices in Wide Runs.}

Let $\mathcal{W}$ denote the area of the disk covered by all widening
sectors.  By \Cref{lem:whp-bad-sectors} the total number of widening
sectors is small, with high probability.  As a consequence,
$\mathcal{W}$ is small as well and we can derive that the size of the
vertex set $V(\mathcal{W})$ containing all vertices in all widening
sectors is sublinear with high probability.

\begin{lemma}
  \label{lem:whp-vertices-wide-runs}
  Let $G$ be a hyperbolic random graph on $n$ vertices. Then, with
  high probability, the number of vertices in wide runs is bounded by
  \begin{align*}
    \vert V(\mathcal{W}) \vert = \mathcal{O} \left( \frac{\tau^{3/4} \cdot n}{\log^{(2)}(n)^{1/4} \cdot \log^{(3)}(n)^{1/2}} \right).
  \end{align*}
\end{lemma}
\begin{proof}
  We start by computing the expected number of vertices in
  $\mathcal{W}$ and show concentration afterwards.  The probability
  for a given vertex to fall into $\mathcal{W}$ is equal to its
  measure~$\mu(\mathcal{W})$.  Since the angular coordinates of the
  vertices are distributed uniformly at random, we have
  $\mu(\mathcal{W}) = W/n'$, where $W$ denotes the number of widening
  sectors and $n'$ is the total number of sectors, which is given by
  \Cref{lem:num-sectors}.  The expected number of vertices in
  $\mathcal{W}$ is then
  \begin{align}
    \label{eq:exp-wide-vertices}
    \mathbb{E}[\vert V(\mathcal{W}) \vert] = n \mu(\mathcal{W}) = n \frac{W}{n'} = \frac{1}{2} W \cdot \gamma(n, \tau) (1 \pm o(1)),
  \end{align}
  where the last equality holds since $1/(1 + x) = 1 - \Theta(x)$ is
  valid for $x = \pm o(1)$.  Note that the number of widening sectors
  $W$ is itself a random variable.  Therefore, we apply the law of
  total expectation and consider different outcomes of $W$ weighted
  with their probabilities.  Motivated by the previously determined
  probabilistic bound on $W$ (\Cref{lem:whp-bad-sectors}), we consider
  the events $W \le g(n)$ as well as $W > g(n)$, where
  \begin{align*}
    g(n) = \frac{c \cdot \tau^{3/4} \cdot n}{\gamma(n, \tau) \cdot \log^{(2)}(n)^{1/4} \cdot \log^{(3)}(n)^{1/2}},
  \end{align*}
  for sufficiently large $c > 0$ and $n$.  With this, we can compute
  the expected number of vertices in $\mathcal{W}$ as
  \begin{align*}
    \mathbb{E}[\vert V(\mathcal{W})\vert ] =~&\mathbb{E}[\vert V(\mathcal{W})\vert  \mid W \le g(n)] \cdot \Pr[W \le g(n)]~+ \\
                                   &\mathbb{E}[\vert V(\mathcal{W})\vert  \mid W > g(n)] \cdot \Pr[W > g(n)].
  \end{align*}
  To bound the first summand, note that $\Pr[W \le g(n)] \le 1$.
  Further, by applying \Cref{eq:exp-wide-vertices} from above, we have
  \begin{align*}
    \mathbb{E}[\vert V(\mathcal{W}) \vert \mid W \le g(n)] \cdot \Pr[W \le g(n)] &\le \frac{1}{2} g(n) \cdot \gamma(n, \tau) (1 \pm o(1)).
  \end{align*}
  In order to bound the second summand, note that $n$ is an obvious
  upper bound on $\mathbb{E}[\vert V(\mathcal{W}) \vert]$.  Moreover,
  by \Cref{lem:whp-bad-sectors} it holds that
  $\Pr[W > g(n)] = \mathcal{O}(n^{-c_1})$ for any $c_1 > 0$.  As a
  result we have
  \begin{align*}
    \mathbb{E}[\vert V(\mathcal{W}) \vert \mid W > g(n)] \cdot \Pr[W > g(n)] \le n \Pr[W > g(n)] &= \mathcal{O}(n^{-c_1 + 1}),
  \end{align*}
  for any $c_1 > 0$.  Clearly, the first summand dominates the second
  and we can conclude that
  $\mathbb{E}[\vert V(\mathcal{W}) \vert] = \mathcal{O}(g(n) \gamma(n,
  \tau))$.  Consequently, for large enough $n$, there exists a
  constant $c_2 > 0$ such that $\hat{g}(n) = c_2 g(n) \gamma(n, \tau)$
  is a valid upper bound on $\mathbb{E}[\vert V(\mathcal{W}) \vert]$.
  This allows us to apply the Chernoff bound in \Cref{col:chernoff} to
  bound the probability that $\vert V(\mathcal{W}) \vert$ exceeds
  $\hat{g}(n)$ by more than a constant factor as
  \begin{align*}
    \Pr[\vert V(\mathcal{W}) \vert \ge (1 + \varepsilon)\hat{g}(n)] \le e^{-\varepsilon^2/3 \SpaciousCdot \hat{g}(n)}.
  \end{align*}
  Finally, since $\hat{g}(n)$ can be simplified as
  \begin{align*}
    \hat{g}(n) = c_2 \cdot \frac{c \cdot \tau^{3/4} \cdot n}{\log^{(2)}(n)^{1/4} \cdot \log^{(3)}(n)^{1/2}},
  \end{align*}
  it is easy to see that $\hat{g}(n) = \omega(\log(n))$ and thus
  $\vert V(\mathcal{W}) \vert = \mathcal{O}(\hat{g}(n))$ holds with
  probability $1 - \mathcal{O}(n^{-c_3})$ for any $c_3 > 0$.
\end{proof}

It remains to bound the number of vertices in large components
contained in narrow runs.


\subsubsection{Narrow Runs}

In the following, we differentiate between \emph{small} and
\emph{large} narrow runs, containing at most and more than $\tau
\log^{(2)}(n)$ vertices, respectively.  As before, we first bound the
expected number of vertices in all large narrow runs and deal with
concentration afterwards.

\subparagraph{Expected Number of Vertices in Large Narrow Runs.}

The straight-forward way to bounding the number of vertices in all
large narrow runs is to consider each sector and count the contained
number of vertices, if the sector is part of a large narrow run.
Unfortunately, whether this is the case depends on the surrounding
sectors and whether they are empty (ending the run) or contain lots of
vertices (to make the run large), and dealing with these stochastic
dependencies is difficult.

To relax these dependencies, we determine an upper bound on the number
of vertices in large narrow runs, by not only considering sectors that
are part of such a run, but also ones that are in the proximity
thereof.  More precisely, for a sector $\mathcal{S}$ we define its
\emph{narrow proximity} $\mathcal{P}_\mathcal{S}$ as $\mathcal{S}$
together with the $w - 1$ sectors to its left and the $w - 1$ sectors
to its right.  If $\mathcal{S}$ is part of a large narrow run, then
there are more than $\tau \log^{(2)}(n)$ vertices in its narrow
proximity.  Note, however, that this condition is not sufficient: Even
if there are as many vertices in $\mathcal{P}_\mathcal{S}$, there
could be empty sectors that cut $\mathcal{S}$ off from the
corresponding sectors, in which case $\mathcal{S}$ is not part of a
large narrow run.

We start by bounding the expected number of vertices in the narrow
proximity of a sector.

\begin{lemma}
  \label{lem:exp-vertices-proximity}
  Let $G$ be a hyperbolic random graph on $n$ vertices, let
  $\mathcal{S}$ be a sector, and let $\mathcal{P}_\mathcal{S}$ be its
  narrow proximity.  Then, 
  $\mathbb{E}[\vert
  V(\mathcal{P}_\mathcal{S}) \vert] \le e^{\gamma(n, \tau)}
  \log^{(3)}(n) \gamma(n, \tau) (1 \pm o(1))$.
\end{lemma}
\begin{proof}
  The narrow proximity of $\mathcal{S}$ consists of $\mathcal{S}$
  together with the $w - 1$ sectors to its left and the $w - 1$ ones
  to its right.  In particular, $\mathcal{P}(\mathcal{S})$ consists of
  at most $2w$ sectors.  Since the angular coordinates of the vertices
  are distributed uniformly at random and since we partitioned the
  disk into $n'$ disjoint sectors of equal width, we can derive an
  upper bound on the expected number of vertices in
  $\mathcal{P}_\mathcal{S}$ as $\mathbb{E}[\vert
  V(\mathcal{P}_\mathcal{S}) \vert] \le n \cdot 2w/n'$.
  As $w = e^{\gamma(n, \tau)}\log^{(3)}(n)$ by definition and $n' =
  2n/\gamma(n, \tau) \cdot (1 \pm o(1))$ according to
  \Cref{lem:num-sectors}, we have
  \begin{align*}
    \mathbb{E}[\vert V(\mathcal{P}_\mathcal{S}) \vert] \le e^{\gamma(n, \tau)} \log^{(3)}(n) \gamma(n, \tau) (1 \pm o(1))^{-1}.
  \end{align*}
  Since $1/(1 + x) = 1 - \Theta(x)$ for $x = \pm o(1)$, we obtain the
  claimed bound.
\end{proof}

Using this upper bound, we can bound the probability that the number
of vertices in the narrow proximity of a sector exceeds the threshold
$\tau \log^{(2)}(n)$ by a certain amount.

\begin{lemma}
  \label{lem:chernoff-vertices-proximity}
  Let $G$ be a hyperbolic random graph on $n$ vertices, let
  $\mathcal{S}$ be a sector, and let $\mathcal{P}_\mathcal{S}$ be its
  narrow proximity.  For $k > \tau \log^{(2)}(n)$ and $n$ large
  enough, it holds that $\Pr[\vert V(\mathcal{P}_\mathcal{S}) \vert =
  k] \le e^{-k/18}$.
\end{lemma}
\begin{proof}
  First note that $\Pr[\vert V(\mathcal{P}_\mathcal{S}) \vert = k] \le
  \Pr[\vert V(\mathcal{P}_\mathcal{S}) \vert \ge k]$.  In order to
  show that $\Pr[\vert V(\mathcal{P}_\mathcal{S}) \vert \ge k]$ is
  small, we aim to apply the Chernoff bound in \Cref{col:chernoff},
  choosing $\varepsilon = 1/2$ and $g(n) = 2/3 \cdot k$ as an upper
  bound on $\mathbb{E}[\vert V(\mathcal{P}_\mathcal{S}) \vert]$.  To
  see that this is a valid choice, we use
  \Cref{lem:exp-vertices-proximity} and substitute $\gamma(n, \tau) =
  \log(\tau \log^{(2)}(n)/(2 \log^{(3)}(n)^2))$, which yields
  \begin{align*}
    \mathbb{E}[\vert V(\mathcal{P}_\mathcal{S}) \vert] &\le e^{\gamma(n, \tau)} \log^{(3)}(n) \gamma(n, \tau) (1 \pm o(1)) \\
                                                       &= \frac{\tau \log^{(2)}(n)}{2 \log^{(3)}(n)^2} \cdot \log^{(3)}(n) \cdot \log \left( \frac{\tau \log^{(2)}(n)}{2 \log^{(3)}(n)^2} \right) (1 \pm o(1)) \\
                                                       &= \frac{\tau \log^{(2)}(n)}{2 \log^{(3)}(n)} \cdot \left( \log^{(3)}(n) - \left(2\log^{(4)}(n) - \log(\tau / 2) \right) \right) (1 \pm o(1)) \\
                                                       &= \frac{1}{2} \cdot \tau \log^{(2)}(n) \cdot \left( 1 - \frac{2\log^{(4)}(n) - \log(\tau / 2)}{\log^{(3)}(n)} \right) (1 \pm o(1)).
  \end{align*}
  Note, that the first error term is equivalent to $(1 - o(1))$ and
  that $(1 \pm o(1)) \le 4/3$ holds for $n$ large enough.
  Consequently, for sufficiently large $n$, we have $\mathbb{E}[\vert
  V(\mathcal{P}_\mathcal{S}) \vert] \le 2/3 \cdot \tau \log^{(2)}(n)$.
  Since $k > \tau \log^{(2)}(n)$, it follows that $g(n) = 2/3 \cdot k$
  is a valid upper bound on $\mathbb{E}[\vert
  V(\mathcal{P}_\mathcal{S}) \vert]$.  Therefore, we can apply the
  Chernoff bound in \Cref{col:chernoff} to conclude that
  \begin{align*}
    \Pr[\vert V(\mathcal{P}_\mathcal{S}) \vert \ge k] &\le e^{-(1/2)^2/3 \SpaciousCdot g(n)} = e^{-1/12 \SpaciousCdot 2/3 \SpaciousCdot k} = e^{-k/18}.
  \end{align*}
\end{proof}

We are now ready to bound the expected value of the number $N$ of
vertices in all large narrow runs.

\begin{lemma}
  \label{lem:exp-large-narrow-runs}
  Let $G$ be a hyperbolic random graph.  Then, the expected number of
  vertices in all large narrow runs is bounded by
  \begin{align*}
    \mathbb{E}[N] = \mathcal{O} \left(\frac{n \cdot \log^{(3)}(n)}{\gamma(n, \tau) \log(n)^{\tau/18}} \right).
  \end{align*}
\end{lemma}
\begin{proof}
  For the $i$-th sector $\mathcal{S}_i$ ($i \in [n']$) we define a
  random variable $N_i$, with $N_i = \vert V(\mathcal{S}_i) \vert$ if
  $\mathcal{S}_i$ is part of a large narrow run and $N_i = 0$
  otherwise.  Then $N = \sum_{i \in [n']} N_i$.  As mentioned above,
  we compute an upper bound of $\mathbb{E}[N]$ by considering random
  variables $N_i'$ instead, where $N_i' = \vert V(\mathcal{S}_i)
  \vert$ if the number of vertices in the narrow proximity of
  $\mathcal{S}_i$ exceeds the threshold $t = \tau \log^{(2)}(n)$, and
  $N_i' = 0$ otherwise.  By the above argumentation it holds that
  $N_i' \ge N_i$ and thus for $N' = \sum_{i \in [n']} N_i'$ we have
  $N' \ge N$.  Consequently, it suffices to show that the claimed
  bound holds for $\mathbb{E}[N']$.  To this end, we compute
  \begin{align*}
    \mathbb{E}[N'] = \sum_{i = 1}^{n'} \mathbb{E}[N_i'] = \sum_{i = 1}^{n'} \sum_{k = 0}^{n} \mathbb{E}\left[N_i~\big\vert~\vert V(\mathcal{P}_{\mathcal{S}_i}) \vert = k \right] \cdot \Pr[\vert V(\mathcal{P}_{\mathcal{S}_i}) \vert = k],
  \end{align*}
  where the second equality is obtained using the law of total
  expectation.  Note that we have $N_i' = 0$ whenever $\vert
  V(\mathcal{P}_{\mathcal{S}_i}) \vert \le t$, and $N_i = \vert
  V(\mathcal{S}_i) \vert$, otherwise.  Thus, the expression simplifies
  to
  \begin{align*}
    \mathbb{E}[N'] = \sum_{i = 1}^{n'} \sum_{k = t + 1}^{n} \mathbb{E}\left[\vert V(\mathcal{S}_i) \vert~\big\vert~\vert V(\mathcal{P}_{\mathcal{S}_i}) \vert = k \right] \cdot \Pr[\vert V(\mathcal{P}_{\mathcal{S}_i}) \vert = k],
  \end{align*}
  
  In each summand we are interested in the expected number of vertices
  in a sector~$\mathcal{S}_i$, conditioned on the fact that its narrow
  proximity contains exactly $k$ vertices.  Since the angular
  coordinates of the vertices are distributed uniformly and the narrow
  proximity consists of $2w - 1$ sectors including $\mathcal{S}_i$,
  the expected number that end up in $\mathcal{S}_i$ is given by $k /
  (2w - 1) \le k/w$.  It follows that 
  \begin{align*}
    \mathbb{E}[N'] \le \sum_{i = 1}^{n'} \sum_{k = t + 1}^{n} \frac{k}{w} \cdot \Pr[\vert V(\mathcal{P}_{\mathcal{S}_i}) \vert = k].
  \end{align*}
  The probability can be bounded using
  Lemma~\ref{lem:chernoff-vertices-proximity}, which yields for $k >
  \tau \log^{(2)}(n) = t$ that $\Pr[\vert
  V(\mathcal{P}_{\mathcal{S}_i}) \vert = k] \le e^{-k/18}$.  Thus,
  \begin{align*}
    \mathbb{E}[N'] \le \frac{n'}{w} \sum_{k = t + 1}^{n} k \cdot e^{-k/18}.
  \end{align*}
  Note that the sum is of the form $\sum k b^k$ for $b = e^{-1/18} <
  1$, which is the derivative of the geometric series multiplied by
  $b$.  Consequently, we obtain an upper bound by bounding the limits
  as $t + 1 > t$ and $n < \infty$ and applying the identity
  \begin{align*}
    \sum_{k = t}^{\infty} kb^k = \frac{b^{t} (t + b - tb)}{(b - 1)^2},
  \end{align*}
  which is valid for $b < 1$ and reduces to $\mathcal{O}(b^t \cdot t)$
  for constant $b$.  Substituting $b = e^{-1/18}$ and $t = \tau
  \log^{(2)}(n)$ then yields
  \begin{align*}
    \mathbb{E}[N'] \le \frac{n'}{w} \mathcal{O}\left(e^{-\tau \log^{(2)}(n) / 18} \cdot \log^{(2)}(n)\right) = \mathcal{O}\left( \frac{n'}{w} \frac{\log^{(2)}(n)}{\log(n)^{\tau/18}} \right).
  \end{align*}
  Finally, since $w = e^{\gamma(n, \tau)}\log^{(3)}(n)$ by definition
  and $n' = \mathcal{O}(n / \gamma(n, \tau))$ by
  \Cref{lem:num-sectors}, where we defined $\gamma(n, \tau) = \log(
  \tau \log^{(2)}(n) / (2 \log^{(3)}(n)^2))$, the above term can be
  simplified to
  \begin{align*}
    \mathbb{E}[N'] &= \mathcal{O}\left(\frac{n}{\gamma(n, \tau) \cdot e^{\gamma(n, \tau)}\log^{(3)}(n)} \cdot \frac{\log^{(2)}(n)}{\log(n)^{\tau / 18}}\right) \\
                   &= \mathcal{O}\left(\frac{n}{\gamma(n, \tau) \cdot \frac{\log^{(2)}(n)}{\log^{(3)}(n)}} \cdot \frac{\log^{(2)}(n)}{\log(n)^{\tau / 18}}\right) \\
                   &= \mathcal{O}\left(\frac{n \cdot \log^{(3)}(n)}{\gamma(n, \tau) \cdot \log(n)^{\tau / 18}} \right).
  \end{align*}
\end{proof}

\subparagraph{Concentration Bound on the Number of Vertices in Large
  Narrow Runs.}

\begin{figure}
  \centering
  \includegraphics{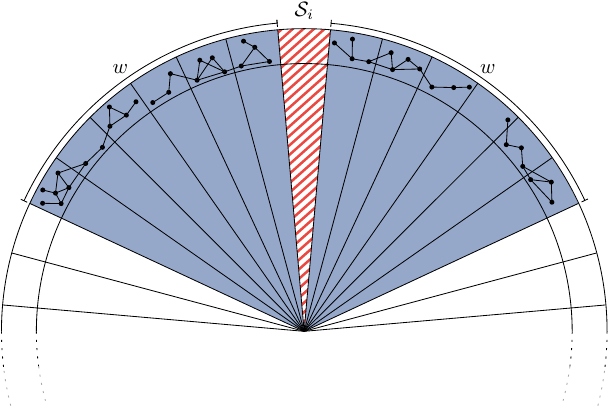}
  \caption{The random variable $S_i$ indicates whether $\mathcal{S}_i$
    contains any vertices.  Changing $S_i$ from $0$ to $1$ or vice
    versa merges two narrow runs or splits a wide run into two narrow
    ones, respectively.  If all vertices were placed in the blue area,
    moving a single vertex in or out of $\mathcal{S}_i$ may change the
    number of vertices in large narrow runs by $n$.}
  \label{fig:bad-sectors-cases}
\end{figure}

To show that the actual number of vertices in large narrow runs $N$ is
not much larger than the expected value, we consider $N$ as a function
of $n$ independent random variables $P_1, \dots, P_n$ representing the
positions of the vertices in the hyperbolic disk.  In order to show
that $N$ does not deviate much from its expected value with high
probability, we would like to apply the method of bounded differences,
which builds on the fact that $N$ satisfies the bounded differences
condition, i.e., that changing the position of a single vertex does
not change $N$ by much.  Unfortunately, this change is not small in
general.

In the worst case, there is a wide run $\mathcal{R}$ that contains all
vertices and a sector $\mathcal{S}_i \subseteq \mathcal{R}$ contains
only one of them.  Moving this vertex out of $\mathcal{S}_i$ may split
the run into two narrow runs (see \Cref{fig:bad-sectors-cases}).
These still contain $n$ vertices, which corresponds to the change in
$N$.  However, this would mean that $\mathcal{R}$ consists of only few
sectors (since it can be split into two narrow runs) and that all
vertices lie within the corresponding (small) area of the disk.  Since
the vertices of the graph are distributed uniformly, this is very
unlikely.  To take advantage of this, we apply the method of
\emph{typical} bounded differences
(\Cref{col:typical-bounded-differences}), which allows us to milden
the effects of the change in the unlikely worst case and to focus on
the typically smaller change of $N$ instead.  Formally, we represent
the typical case using an event $A$ denoting that each run of length
at most $2w + 1$ contains at most $\mathcal{O}(\log(n))$ vertices.  In
the following, we show that $A$ occurs with probability
$1 - \mathcal{O}(n^{-c})$ for any constant $c$, which shows that the
atypical case is very unlikely.

\begin{lemma}
\label{lem:whp-small-splittable-runs}
Let $G$ be a hyperbolic random graph.  Then, each run of length at
most $2w + 1$ contains at most $\mathcal{O}(\log(n))$ vertices with
probability $1 - \mathcal{O}(n^{-c})$ for any constant $c$.
\end{lemma}
\begin{proof}
  We show that the probability for a single run $\mathcal{R}$ of at
  most $2w + 1$ sectors to contain more then $\mathcal{O}(\log(n))$
  vertices is $\mathcal{O}(n^{-c_1})$ for any constant $c_1$.  Since
  there are at most $n' = \mathcal{O}(n)$ runs, applying the union
  bound and choosing $c_1 = c + 1$ then yields the claim.

  Recall that we divided the disk into $n'$ sectors of equal width.
  Since the angular coordinates of the vertices are distributed
  uniformly at random, the probability for a given vertex to a lie in
  $\mathcal{R}$ (i.e., to be in $V(\mathcal{R})$) is given by
  \begin{align*}
    \mu(\mathcal{R}) \le \frac{2w + 1}{n'} = \frac{2 e^{\gamma(n, \tau)}\log^{(3)}(n) + 1}{n'}.
  \end{align*}
  By \Cref{lem:num-sectors} the total number of sectors is given as
  $n' = 2n / \gamma(n, \tau) \cdot (1 \pm o(1))$.  Consequently, we
  can compute the expected number of vertices in $\mathcal{R}$ as
  \begin{align*}
    \mathbb{E}[\vert V(\mathcal{R}) \vert] \le n \mu(\mathcal{R}) = \left(e^{\gamma(n, \tau)} \log^{(3)}(n) + 1/2 \right) \gamma(n, \tau) (1 \pm o(1)).
  \end{align*}
  Substituting
  $\gamma(n, \tau) =
  \mathcal{O}(\log(\log^{(2)}(n)/\log^{(3)}(n)^2))$, we can derive
  that
  \begin{align*}
    \mathbb{E}[\vert V(\mathcal{R}) \vert] \le \mathcal{O}\left(\frac{\log^{(2)}(n)}{\log^{(3)}(n)^2} \log^{(3)}(n) \cdot \log \left(\frac{\log^{(2)}(n)}{\log^{(3)}(n)^2} \right) \right) = \mathcal{O}(\log^{(2)}(n)).
  \end{align*}
  Consequently, it holds that $g(n) = c_2 \log(n)$ is a valid upper
  bound for any $c_2 > 0$ and large enough $n$.  Therefore, we can
  apply the Chernoff bound in \Cref{col:chernoff} to conclude that the
  probability for the number of vertices in $\mathcal{R}$ to exceed
  $g(n)$ is at most
  \begin{align*}
    \Pr[\vert V(\mathcal{R}) \vert \ge (1 + \varepsilon)g(n)] \le e^{-\varepsilon^2/3 \SpaciousCdot g(n)} = n^{-c_2 \varepsilon^2/3}.
  \end{align*}
  Thus, $c_2$ can be chosen sufficiently large, such that
  \begin{align*}
    \Pr[\vert V(\mathcal{R}) \vert \ge (1 + \varepsilon)g(n)] = \mathcal{O}(n^{-c_1})
  \end{align*}
  for any constant $c_1$.
\end{proof}

The method of typical bounded differences now allows us to focus on
this case and to milden the impact of the worst case changes as they
occur with small probability.  Consequently, we can show that the
number of vertices in large narrow runs is sublinear with high
probability.

\begin{lemma}
  \label{lem:whp-large-narrow-runs}
  Let $G$ be a hyperbolic random graph.  Then, with high probability,
  the number of vertices in large narrow runs is bounded by
  \begin{align*}
    N = \mathcal{O} \left(\frac{n \cdot \log^{(3)}(n)}{\gamma(n, \tau) \log(n)^{\tau/18}} \right).
  \end{align*}
\end{lemma}
\begin{proof}
  Recall that the expected number of vertices in all large narrow runs
  is given by \Cref{lem:exp-large-narrow-runs}.  Consequently, we can
  choose $c > 0$ large enough, such that for sufficiently large~$n$ we
  obtain a valid upper bound on $\mathbb{E}[N]$ by choosing
  \begin{align*}
    g(n) = \frac{c \cdot n \cdot \log^{(3)}(n)}{\gamma(n, \tau) \log(n)^{\tau/18}}.
  \end{align*}
  In order to show that $N$ does not exceed $g(n)$ by more than a
  constant factor with high probability, we apply the method of
  typical bounded differences
  (\Cref{col:typical-bounded-differences}).  To this end, we consider
  the typical event $A$, denoting that each run of at most $2w + 1$
  sectors contains at most $\mathcal{O}(\log(n))$ vertices, and it
  remains to determine the parameters $\Delta_i^A \le \Delta_i$ with
  which $N$ satisfies the typical bounded differences condition with
  respect to $A$ (see
  \Cref{eq:typical-bounded-differences-condition}).  Formally, we have
  to show that for all $i \in \{1, \dots, n\}$
  \begin{align*}
    \vert N(P_1, \dots, P_i, \dots, P_n) - N(P_1, \dots, P_i', \dots, P_n) \vert \le 
    \begin{cases}
      \Delta_i^A,\mkern-16mu&\text{if}~(P_1, \dots, P_i, \dots, P_n) \in A,\\
      \Delta_i,\mkern-16mu&\text{otherwise}.
    \end{cases}
  \end{align*}
  As argued before, changing the position $P_i$ of vertex $i$ to
  $P_i'$ may result in a change of $n$ in the worst case.  Therefore,
  $\Delta_i = n$ is a valid bound for all $i \in [n]$.  To bound the
  $\Delta_i^A$, we have to consider the following situation.  We start
  with a set of positions such that all runs of $2w + 1$ sectors
  contain at most $\mathcal{O}(\log(n))$ vertices and we want to bound
  the change in $N$ when changing the position $P_i$ of a single
  vertex $i$.  In this case, splitting a wide run or merging two
  narrow runs can only change $N$ by $\mathcal{O}(\log(n))$.
  Consequently, we can choose $\Delta_i^A = \mathcal{O}(\log(n))$ for
  all $i \in [n]$.  By \Cref{col:typical-bounded-differences} we can
  now bound the probability that $N$ exceeds $g(n)$ by more than a
  constant factor $c_1$ as
  \begin{align*}
    \Pr[N > c_1 g(n)] \le e^{-((c_1 - 1) g(n))^2 / (2 \Delta)} + \Pr[\bar{A}] \cdot \sum_{i \in [n]} 1 / \varepsilon_i,
  \end{align*}
  for any $\varepsilon_1, \dots, \varepsilon_n \in (0, 1]$ and
  $\Delta = \sum_{i \in [n]} (\Delta_i^A + \varepsilon_i (\Delta_i -
  \Delta_i^A))^2$.  By substituting the previously determined
  $\Delta_i^A$ and $\Delta_i$, as well as, choosing
  $\varepsilon_i = 1/n$ for all $i \in [n]$, we obtain
  \begin{align*}
    \Delta = \mathcal{O} \left( n \cdot \left(\log(n) + 1/n \cdot (n - \log(n)) \right)^2 \right) = \mathcal{O}(n \cdot \log(n)^2).
  \end{align*}
  Thus,
  \begin{align*}
    &\Pr[N > c_1 g(n)] \\
    &\quad\le \exp \left(- \Theta \left(n^2 \cdot \left( \frac{\log^{(3)}(n)}{\gamma(n, \tau) \log(n)^{\tau/18}} \right)^2 \right) \cdot \frac{1}{\mathcal{O}(n \log(n)^2)} \right) + \Pr[\bar{A}] \cdot \sum_i 1 / \varepsilon_i \\
    &\quad= \exp \left(- \Omega \left( n \cdot \left(\frac{\log^{(3)}(n)}{\gamma(n, \tau) \log(n)^{1 + \tau/18}} \right)^2 \right) \right) + \Pr[\bar{A}] \cdot \sum_i 1 / \varepsilon_i \\
    &\quad= \exp \left( -\Omega \left( n \cdot \left(\frac{1}{\log(n)^{1 + \tau/18}} \right)^2 \right) \right) + \Pr[\bar{A}] \cdot \sum_i 1 / \varepsilon_i,
  \end{align*}
  where the last equality holds, since
  $\gamma(n, \tau) = \mathcal{O}(\log^{(3)}(n))$.  By further
  simplifying the exponent, we can derive that the first part of the
  sum is $\exp(-\omega(\log(n)))$.  It follows that
  $\Pr[N > c_1 g(n)] \le n^{-c_2} + \Pr[\bar{A}] \cdot \sum_{i \in
    [n]} 1 / \varepsilon_i$ holds for any $c_2 > 0$ and sufficiently
  large $n$.  It remains to bound the second part of the sum.  Since
  $\varepsilon_i = 1/n$ for all $i \in [n]$, we have
  $\Pr[\bar{A}] \cdot \sum_{i \in [n]} 1 / \varepsilon_i =
  \Pr[\bar{A}] \cdot n^2$.  By \Cref{lem:whp-small-splittable-runs} it
  holds that $\Pr[\bar{A}] = \mathcal{O}(n^{-c_3})$ for any~$c_3$.
  Consequently, we can choose $c_3$ such that
  $\Pr[\bar{A}] \cdot n^{2} = \mathcal{O}(n^{-(c_3 - 2)})$ for any
  $c_3$, which concludes the proof.
\end{proof}

\subsection{The Complete Disk}
\label{sec:complete-disk}

In the previous subsections we determined the number of vertices that
are greedily added to the vertex cover in the inner disk and outer
band, respectively.  Before proving our main theorem, we are now ready
to prove a slightly stronger version that shows how the
parameter~$\tau$ can be used to obtain a trade-off between
approximation performance and running time.

\begin{theorem}
  \label{thm:vertex-cover-approximation-trade-off}
  Let $G$ be a hyperbolic random graph on $n$ vertices with power-law
  exponent $\beta = 2\alpha + 1$ and let $\tau > 0$ be constant.
  Given the radii of the vertices, an approximate vertex cover of $G$
  can be computed in time $\mathcal{O}(n \log(n) + m \log(n)^{\tau})$,
  such that the approximation factor is $(1 + \mathcal{O}(\gamma(n,
  \tau)^{-\alpha}))$ asymptotically almost surely.
\end{theorem}
\begin{proof}

  \emph{Running Time.}  We start by sorting the vertices of the graph
  in order of increasing radius, which can be done in time
  $\mathcal{O}(n \log(n))$.  Afterwards, we iterate them and perform
  the following steps for each encountered vertex $v$.  We add $v$ to
  the cover, remove it from the graph, and identify connected
  components of size at most $\tau \log\log(n)$ that were separated by
  the removal.  The first two steps can be performed in time
  $\mathcal{O}(1)$ and $\mathcal{O}(\deg(v))$, respectively.
  Identifying and solving small components is more involved.
  Removing~$v$ can split the graph into at most $\deg(v)$ components,
  each containing a neighbor $u$ of $v$.  Such a component can be
  identified by performing a breadth-first search (BFS) starting at
  $u$.  Each BFS can be stopped as soon as it encounters more than
  $\tau \log\log(n)$ vertices.  The corresponding subgraph contains at
  most $(\tau \log\log(n))^2$ edges.  Therefore, a single BFS takes
  time $\mathcal{O}(\log\log(n)^2)$.  Whenever a component of size at
  most $n_c = \tau \log\log(n)$ is found, we compute a minimum vertex
  cover for it in time $1.1996^{n_c} \cdot
  n_c^{\mathcal{O}(1)}$~\cite{xn-eamis-17}. Since
  $n_c^{\mathcal{O}(1)} = \mathcal{O}((e / 1.1996)^{n_c})$, this
  running time is bounded by $\mathcal{O}(e^{n_c}) =
  \mathcal{O}(\log(n)^\tau)$.  Consequently, the time required to
  process each neighbor of $v$ is $\mathcal{O}(\log(n)^{\tau})$.
  Since this is potentially performed for all neighbors of $v$, the
  running time of this third step can be bounded by introducing an
  additional factor of $\deg(v)$.  We then obtain the total running
  time $T(n, m, \tau)$ of the algorithm by taking the time for the
  initial sorting and adding the sum of the running times of the above
  three steps over all vertices, which yields
  \begin{align*}
    T(n, m, \tau) &= \mathcal{O}(n \log(n)) + \sum_{v \in V} \left(\mathcal{O}(1) + \mathcal{O}(\deg(v)) + \deg(v) \cdot \mathcal{O}(\log(n)^{\tau}) \right)\\
                  &= \mathcal{O}(n \log(n)) + \mathcal{O} \Bigg(\log(n)^{\tau} \cdot \sum_{v \in V} \deg(v) \Bigg) \\
                  &= \mathcal{O}(n \log(n) + m \log(n)^{\tau}).
  \end{align*}

  \emph{Approximation Ratio.}  As argued before, we obtain a valid
  vertex cover for the whole graph, if we take all vertices in
  $V_{\text{Greedy}}$ together with a vertex cover $C_{\text{Exact}}$
  of $G[V_{\text{Exact}}]$.  The approximation ratio of the resulting
  cover is then given by the quotient
  \begin{align*}
    \delta = \frac{\vert V_{\text{Greedy}}\vert + \vert C_{\text{Exact}}\vert }{\vert C_{\text{OPT}}\vert },
  \end{align*}
  where $C_{\text{OPT}}$ denotes an optimal solution.  Since all
  components in $G[V_{\text{Exact}}]$ are solved optimally and since
  any minimum vertex cover for the whole graph induces a vertex cover
  on $G[V']$ for any vertex subset $V' \subseteq V$, it holds that
  $\vert C_{\text{Exact}}\vert \le \vert C_{\text{OPT}}\vert $.
  Therefore, the approximation ratio can be bounded by
  $\delta \le 1 + \vert V_{\text{Greedy}}\vert /\vert
  C_{\text{OPT}}\vert $.
  To bound the number of vertices in $V_{\text{Greedy}}$, we add the
  number of vertices $I$ in the inner disk $\mathcal{I}$, as well as the
  numbers of vertices $W$ in the outer band that are contained in wide
  runs and the number of vertices $N$ that are contained in large
  narrow runs.  That is,
  \begin{align*}
    \delta \le 1 + \frac{I + W + N }{\vert C_{OPT}\vert }.
  \end{align*}
  Upper bounds on $I$, $W$, and $N$ that hold with high probability
  are given by
  \Cref{lem:inner-disk-bad-nodes-whp,lem:whp-vertices-wide-runs,lem:whp-large-narrow-runs},
  respectively.  Furthermore, it was previously shown that the size of
  a minimum vertex cover on a hyperbolic random graph is $\vert
  C_{OPT}\vert =~\Omega(n)$, asymptotically almost
  surely~\cite[Theorems 4.10 and 5.8]{cfr-ggdsfn-16}.  We obtain
  \begin{align*}
    \delta &= 1 + \mathcal{O} \left( \frac{1}{\gamma(n, \tau)^{\alpha}} + \frac{1}{\log^{(2)}(n)^{1/4} \cdot \log^{(3)}(n)^{1/2}} + \frac{\log^{(3)}(n)}{\gamma(n, \tau) \log(n)^{\tau/18}} \right).
  \end{align*}
  Since $\gamma(n, \tau) = \mathcal{O}(\log^{(3)}(n))$, the first
  summand dominates asymptotically.
\end{proof}

\begin{backInTime}{thm-vertex-cover-efficient-approximation}
  \begin{theorem}
    Let $G$ be a hyperbolic random graph on $n$ vertices.  Given the
    radii of the vertices, an approximate vertex cover of $G$ can be
    computed in time $\mathcal{O}(m \log(n))$, such that the
    approximation ratio is $(1 + o(1))$ asymptotically almost surely.
  \end{theorem}
  \begin{proof}
    By \Cref{thm:vertex-cover-approximation-trade-off} we can compute
    an approximate vertex cover in time
    $\mathcal{O}(n \log(n) + m \log(n)^{\tau})$, such that the
    approximation factor is
    $1 + \mathcal{O}(\gamma(n, \tau)^{-\alpha})$, asymptotically
    almost surely.  By choosing $\tau = 1$ we get
    $\gamma(n, 1) = \omega(1)$, which yields an approximation factor
    of $(1 + o(1))$, since $\alpha \in (1/2, 1)$.  Additionally, the
    bound on the running time can be simplified to
    $\mathcal{O}(n \log(n) + m \log(n))$.  The claim then follows
    since we assume the graph to be connected, which implies that the
    number of edges is $m = \Omega(n) $.
  \end{proof}
\end{backInTime}

\section{Experimental Evaluation}
\label{sec:experiments}

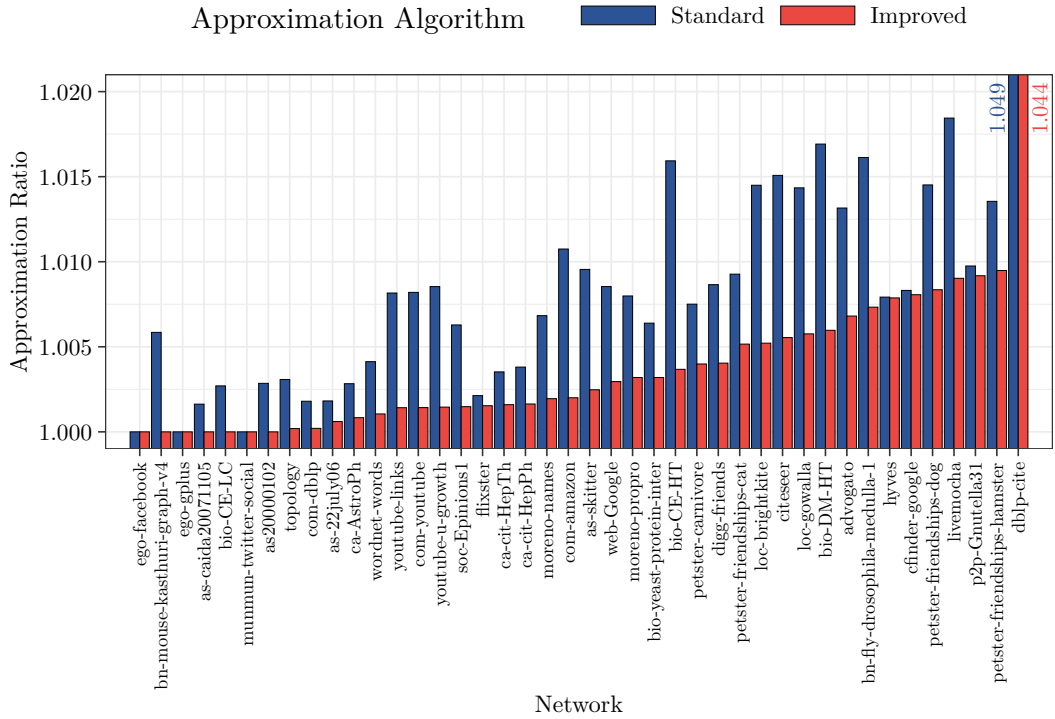
\begin{figure}[t]
  \centering
  \input{plot.tex}
  \caption{Approximation ratios obtained using the standard greedy
    approach (blue) and our improved version (red) on a selection of
    real-world networks.  The parameter adjusting the component size
    threshold was chosen as $\tau = 10$.  For the sake of readability
    the bars denoting the ratios for the \texttt{dblp-cite} network
    were cropped and the actual values written next to them.}
    \label{fig:experiments}
\end{figure}

It remains to evaluate how well the predictions of our analysis on
hyperbolic random graphs translate to real-world networks.  According
to the model, vertices near the center of the disk can likely be added
to the vertex cover safely, while vertices near the boundary need to
be treated more carefully (see \Cref{sec:algorithm}).  Moreover, it
predicts that these boundary vertices can be found by identifying
small components that are separated when removing vertices near the
center.  Due to the correlation between the radii of the vertices and
their degrees~\cite{gpp-rhg-12}, this points to a natural extension of
the standard greedy approach: While iteratively adding the vertex with
the largest degree to the cover, small separated components are solved
optimally.  To evaluate how this approach compares to the standard
greedy algorithm, we measured the approximation ratios on the largest
connected component of a selection of 42 real-world networks from
several network datasets~\cite{k-k-13, ra-ndrigav-15}.  The results of
our empirical analysis are summarized in \Cref{fig:experiments}.

Our experiments confirm that the standard greedy approach already
yields close to optimal approximation ratios on all networks, as
observed previously~\cite{dgd-vccn-13}.  In fact, the ``worst''
approximation ratio is only $1.049$ for the network
\texttt{dblp-cite}.  The average lies at just $1.009$.

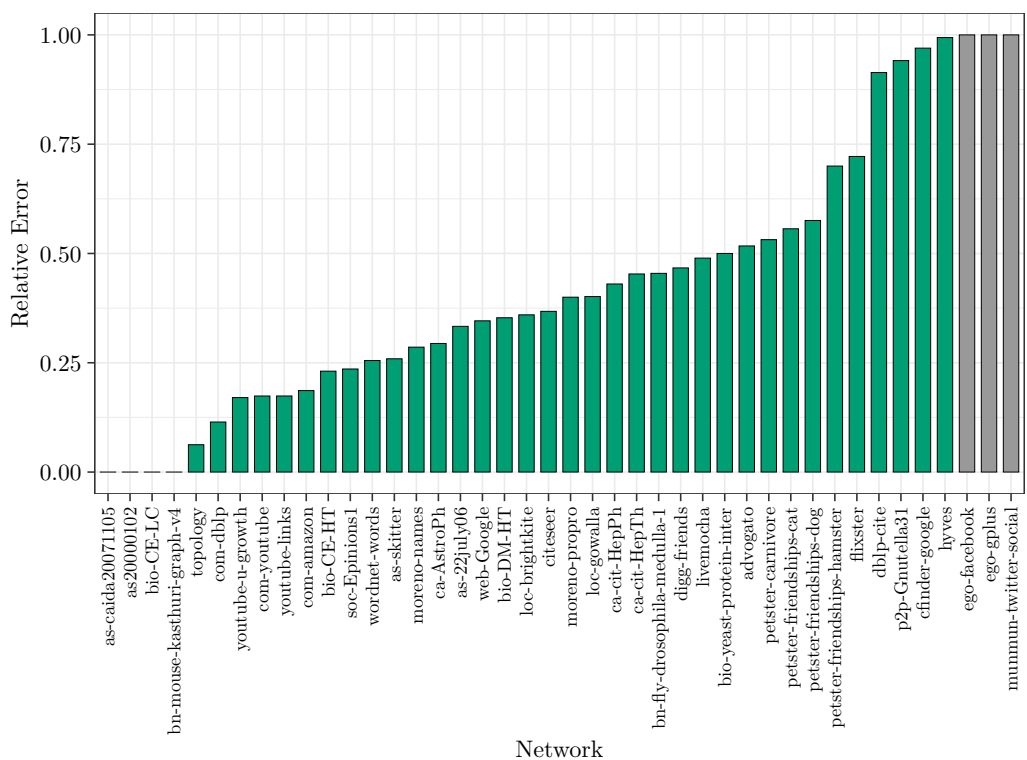
\begin{figure}[t]
  \centering
  \input{plot_relative.tex}
  \caption{Relative error of the improved greedy compared to the
    standard approach.  The parameter adjusting the component size
    threshold was chosen as $\tau = 10$.  Gray bars indicate that no
    error could be determined since the standard approach found an
    optimum already.}
  \label{fig:experiments-relative}
\end{figure}

Clearly, our adapted greedy approach performs at least as well as the
standard greedy algorithm.  In fact, for $\tau = 1$ the sizes of the
components that are solved optimally on the considered networks are at
most $3$.  For components of this size the standard greedy approach
performs optimally.  Therefore, the approximation performances of the
standard and the adapted greedy match in this case.  However, the
adapted greedy algorithm allows for improving the approximation ratio
by increasing the size of the components that are solved optimally.
In our experiments, we chose $10 \lceil \log\log(n) \rceil$ as the
component size threshold, which corresponds to setting $\tau = 10$.
The resulting impact can be seen in \Cref{fig:experiments-relative},
which shows the error of the adapted greedy compared to the one of the
standard greedy algorithm.  This relative error is measured as the
fraction of the number of vertices by which the adapted greedy and the
standard approach exceed an optimum solution.  That is, a relative
error of $0.5$ indicates that the adapted greedy halved the number of
vertices by which the solution of the standard greedy exceeded an
optimum.  Moreover, a relative error of $0$ indicates that the adapted
greedy found an optimum when the standard greedy did not.  The
relative error is omitted (gray in \Cref{fig:experiments-relative}) if
the standard greedy already found an optimum, i.e., there was no error
to improve on.  For more than $69\%$ of the considered networks (29
out of 42) the relative error is at most $0.5$ and the average
relative error is $0.39$.  Since the behavior of the two algorithms
only differs when it comes to small separated components, this
indicates that the predictions of the model that led to the
improvement of the standard greedy approach do translate to real-world
networks.  In fact, the average approximation ratio obtained using the
standard greedy algorithm is reduced from $1.009$ to $1.004$ when
using the adapted greedy approach.

\clearpage



\bibliography{hyperbolic_vertex_cover_approximation}

\end{document}

%% file: plot.tex
\begin{tikzpicture}[x=1pt,y=1pt]
\definecolor{fillColor}{RGB}{255,255,255}
\path[use as bounding box,fill=fillColor,fill opacity=0.00] (0,0) rectangle (339.67,390.26);
\begin{scope}
\path[clip] (  1.84,  0.00) rectangle (337.83,390.26);
\definecolor{drawColor}{RGB}{255,255,255}
\definecolor{fillColor}{RGB}{255,255,255}

\path[draw=drawColor,line width= 0.6pt,line join=round,line cap=round,fill=fillColor] (  1.84,  0.00) rectangle (337.83,390.26);
\end{scope}
\begin{scope}
\path[clip] (102.94, 21.83) rectangle (337.83,374.16);
\definecolor{fillColor}{RGB}{255,255,255}

\path[fill=fillColor] (102.94, 21.83) rectangle (337.83,374.16);
\definecolor{drawColor}{gray}{0.92}

\path[draw=drawColor,line width= 0.3pt,line join=round] (140.31, 21.83) --
	(140.31,374.16);

\path[draw=drawColor,line width= 0.3pt,line join=round] (193.70, 21.83) --
	(193.70,374.16);

\path[draw=drawColor,line width= 0.3pt,line join=round] (247.08, 21.83) --
	(247.08,374.16);

\path[draw=drawColor,line width= 0.3pt,line join=round] (300.46, 21.83) --
	(300.46,374.16);

\path[draw=drawColor,line width= 0.6pt,line join=round] (102.94, 34.88) --
	(337.83, 34.88);

\path[draw=drawColor,line width= 0.6pt,line join=round] (102.94, 43.03) --
	(337.83, 43.03);

\path[draw=drawColor,line width= 0.6pt,line join=round] (102.94, 51.19) --
	(337.83, 51.19);

\path[draw=drawColor,line width= 0.6pt,line join=round] (102.94, 59.35) --
	(337.83, 59.35);

\path[draw=drawColor,line width= 0.6pt,line join=round] (102.94, 67.50) --
	(337.83, 67.50);

\path[draw=drawColor,line width= 0.6pt,line join=round] (102.94, 75.66) --
	(337.83, 75.66);

\path[draw=drawColor,line width= 0.6pt,line join=round] (102.94, 83.81) --
	(337.83, 83.81);

\path[draw=drawColor,line width= 0.6pt,line join=round] (102.94, 91.97) --
	(337.83, 91.97);

\path[draw=drawColor,line width= 0.6pt,line join=round] (102.94,100.12) --
	(337.83,100.12);

\path[draw=drawColor,line width= 0.6pt,line join=round] (102.94,108.28) --
	(337.83,108.28);

\path[draw=drawColor,line width= 0.6pt,line join=round] (102.94,116.44) --
	(337.83,116.44);

\path[draw=drawColor,line width= 0.6pt,line join=round] (102.94,124.59) --
	(337.83,124.59);

\path[draw=drawColor,line width= 0.6pt,line join=round] (102.94,132.75) --
	(337.83,132.75);

\path[draw=drawColor,line width= 0.6pt,line join=round] (102.94,140.90) --
	(337.83,140.90);

\path[draw=drawColor,line width= 0.6pt,line join=round] (102.94,149.06) --
	(337.83,149.06);

\path[draw=drawColor,line width= 0.6pt,line join=round] (102.94,157.21) --
	(337.83,157.21);

\path[draw=drawColor,line width= 0.6pt,line join=round] (102.94,165.37) --
	(337.83,165.37);

\path[draw=drawColor,line width= 0.6pt,line join=round] (102.94,173.53) --
	(337.83,173.53);

\path[draw=drawColor,line width= 0.6pt,line join=round] (102.94,181.68) --
	(337.83,181.68);

\path[draw=drawColor,line width= 0.6pt,line join=round] (102.94,189.84) --
	(337.83,189.84);

\path[draw=drawColor,line width= 0.6pt,line join=round] (102.94,197.99) --
	(337.83,197.99);

\path[draw=drawColor,line width= 0.6pt,line join=round] (102.94,206.15) --
	(337.83,206.15);

\path[draw=drawColor,line width= 0.6pt,line join=round] (102.94,214.30) --
	(337.83,214.30);

\path[draw=drawColor,line width= 0.6pt,line join=round] (102.94,222.46) --
	(337.83,222.46);

\path[draw=drawColor,line width= 0.6pt,line join=round] (102.94,230.62) --
	(337.83,230.62);

\path[draw=drawColor,line width= 0.6pt,line join=round] (102.94,238.77) --
	(337.83,238.77);

\path[draw=drawColor,line width= 0.6pt,line join=round] (102.94,246.93) --
	(337.83,246.93);

\path[draw=drawColor,line width= 0.6pt,line join=round] (102.94,255.08) --
	(337.83,255.08);

\path[draw=drawColor,line width= 0.6pt,line join=round] (102.94,263.24) --
	(337.83,263.24);

\path[draw=drawColor,line width= 0.6pt,line join=round] (102.94,271.39) --
	(337.83,271.39);

\path[draw=drawColor,line width= 0.6pt,line join=round] (102.94,279.55) --
	(337.83,279.55);

\path[draw=drawColor,line width= 0.6pt,line join=round] (102.94,287.71) --
	(337.83,287.71);

\path[draw=drawColor,line width= 0.6pt,line join=round] (102.94,295.86) --
	(337.83,295.86);

\path[draw=drawColor,line width= 0.6pt,line join=round] (102.94,304.02) --
	(337.83,304.02);

\path[draw=drawColor,line width= 0.6pt,line join=round] (102.94,312.17) --
	(337.83,312.17);

\path[draw=drawColor,line width= 0.6pt,line join=round] (102.94,320.33) --
	(337.83,320.33);

\path[draw=drawColor,line width= 0.6pt,line join=round] (102.94,328.48) --
	(337.83,328.48);

\path[draw=drawColor,line width= 0.6pt,line join=round] (102.94,336.64) --
	(337.83,336.64);

\path[draw=drawColor,line width= 0.6pt,line join=round] (102.94,344.79) --
	(337.83,344.79);

\path[draw=drawColor,line width= 0.6pt,line join=round] (102.94,352.95) --
	(337.83,352.95);

\path[draw=drawColor,line width= 0.6pt,line join=round] (102.94,361.11) --
	(337.83,361.11);

\path[draw=drawColor,line width= 0.6pt,line join=round] (102.94,369.26) --
	(337.83,369.26);

\path[draw=drawColor,line width= 0.6pt,line join=round] (113.62, 21.83) --
	(113.62,374.16);

\path[draw=drawColor,line width= 0.6pt,line join=round] (167.00, 21.83) --
	(167.00,374.16);

\path[draw=drawColor,line width= 0.6pt,line join=round] (220.39, 21.83) --
	(220.39,374.16);

\path[draw=drawColor,line width= 0.6pt,line join=round] (273.77, 21.83) --
	(273.77,374.16);

\path[draw=drawColor,line width= 0.6pt,line join=round] (327.15, 21.83) --
	(327.15,374.16);
\definecolor{drawColor}{RGB}{0,0,0}
\definecolor{fillColor}{RGB}{44,83,150}

\path[draw=drawColor,line width= 0.3pt,line cap=rect,fill=fillColor] (-10562.91,316.66) rectangle (113.62,320.33);

\path[draw=drawColor,line width= 0.3pt,line cap=rect,fill=fillColor] (-10562.91, 31.21) rectangle (632.50, 34.88);

\path[draw=drawColor,line width= 0.3pt,line cap=rect,fill=fillColor] (-10562.91,243.26) rectangle (180.71,246.93);

\path[draw=drawColor,line width= 0.3pt,line cap=rect,fill=fillColor] (-10562.91,300.35) rectangle (132.85,304.02);

\path[draw=drawColor,line width= 0.3pt,line cap=rect,fill=fillColor] (-10562.91,137.23) rectangle (212.59,140.90);

\path[draw=drawColor,line width= 0.3pt,line cap=rect,fill=fillColor] (-10562.91,186.17) rectangle (204.81,189.84);

\path[draw=drawColor,line width= 0.3pt,line cap=rect,fill=fillColor] (-10562.91,324.81) rectangle (176.06,328.48);

\path[draw=drawColor,line width= 0.3pt,line cap=rect,fill=fillColor] (-10562.91,169.86) rectangle (198.90,173.53);

\path[draw=drawColor,line width= 0.3pt,line cap=rect,fill=fillColor] (-10562.91, 55.68) rectangle (310.55, 59.35);

\path[draw=drawColor,line width= 0.3pt,line cap=rect,fill=fillColor] (-10562.91,161.70) rectangle (283.71,165.37);

\path[draw=drawColor,line width= 0.3pt,line cap=rect,fill=fillColor] (-10562.91,145.39) rectangle (206.00,149.06);

\path[draw=drawColor,line width= 0.3pt,line cap=rect,fill=fillColor] (-10562.91,112.77) rectangle (266.76,116.44);

\path[draw=drawColor,line width= 0.3pt,line cap=rect,fill=fillColor] (-10562.91,129.08) rectangle (268.40,132.75);

\path[draw=drawColor,line width= 0.3pt,line cap=rect,fill=fillColor] (-10562.91,259.57) rectangle (201.17,263.24);

\path[draw=drawColor,line width= 0.3pt,line cap=rect,fill=fillColor] (-10562.91,202.48) rectangle (228.36,206.15);

\path[draw=drawColor,line width= 0.3pt,line cap=rect,fill=fillColor] (-10562.91, 63.83) rectangle (268.59, 67.50);

\path[draw=drawColor,line width= 0.3pt,line cap=rect,fill=fillColor] (-10562.91,284.03) rectangle (143.84,287.71);

\path[draw=drawColor,line width= 0.3pt,line cap=rect,fill=fillColor] (-10562.91,210.63) rectangle (186.53,214.30);

\path[draw=drawColor,line width= 0.3pt,line cap=rect,fill=fillColor] (-10562.91, 71.99) rectangle (202.39, 75.66);

\path[draw=drawColor,line width= 0.3pt,line cap=rect,fill=fillColor] (-10562.91,332.97) rectangle (113.62,336.64);

\path[draw=drawColor,line width= 0.3pt,line cap=rect,fill=fillColor] (-10562.91,341.12) rectangle (131.01,344.79);

\path[draw=drawColor,line width= 0.3pt,line cap=rect,fill=fillColor] (-10562.91,349.28) rectangle (142.48,352.95);

\path[draw=drawColor,line width= 0.3pt,line cap=rect,fill=fillColor] (-10562.91,235.10) rectangle (136.35,238.77);

\path[draw=drawColor,line width= 0.3pt,line cap=rect,fill=fillColor] (-10562.91, 80.14) rectangle (198.19, 83.81);

\path[draw=drawColor,line width= 0.3pt,line cap=rect,fill=fillColor] (-10562.91, 96.45) rectangle (254.10,100.12);

\path[draw=drawColor,line width= 0.3pt,line cap=rect,fill=fillColor] (-10562.91,275.88) rectangle (157.67,279.55);

\path[draw=drawColor,line width= 0.3pt,line cap=rect,fill=fillColor] (-10562.91,251.41) rectangle (204.78,255.08);

\path[draw=drawColor,line width= 0.3pt,line cap=rect,fill=fillColor] (-10562.91,267.72) rectangle (200.76,271.39);

\path[draw=drawColor,line width= 0.3pt,line cap=rect,fill=fillColor] (-10562.91,120.92) rectangle (274.58,124.59);

\path[draw=drawColor,line width= 0.3pt,line cap=rect,fill=fillColor] (-10562.91,104.61) rectangle (294.22,108.28);

\path[draw=drawColor,line width= 0.3pt,line cap=rect,fill=fillColor] (-10562.91,218.79) rectangle (154.27,222.46);

\path[draw=drawColor,line width= 0.3pt,line cap=rect,fill=fillColor] (-10562.91,357.44) rectangle (113.62,361.11);

\path[draw=drawColor,line width= 0.3pt,line cap=rect,fill=fillColor] (-10562.91,308.50) rectangle (146.49,312.17);

\path[draw=drawColor,line width= 0.3pt,line cap=rect,fill=fillColor] (-10562.91,194.32) rectangle (215.57,197.99);

\path[draw=drawColor,line width= 0.3pt,line cap=rect,fill=fillColor] (-10562.91,292.19) rectangle (133.02,295.86);

\path[draw=drawColor,line width= 0.3pt,line cap=rect,fill=fillColor] (-10562.91, 88.30) rectangle (285.82, 91.97);

\path[draw=drawColor,line width= 0.3pt,line cap=rect,fill=fillColor] (-10562.91, 47.52) rectangle (217.71, 51.19);

\path[draw=drawColor,line width= 0.3pt,line cap=rect,fill=fillColor] (-10562.91,226.95) rectangle (151.26,230.62);

\path[draw=drawColor,line width= 0.3pt,line cap=rect,fill=fillColor] (-10562.91,365.59) rectangle (144.07,369.26);

\path[draw=drawColor,line width= 0.3pt,line cap=rect,fill=fillColor] (-10562.91,178.01) rectangle (181.84,181.68);

\path[draw=drawColor,line width= 0.3pt,line cap=rect,fill=fillColor] (-10562.91, 39.36) rectangle (258.29, 43.03);

\path[draw=drawColor,line width= 0.3pt,line cap=rect,fill=fillColor] (-10562.91,153.54) rectangle (193.76,157.21);
\definecolor{fillColor}{RGB}{234,72,64}

\path[draw=drawColor,line width= 0.3pt,line cap=rect,fill=fillColor] (-10562.91,320.33) rectangle (113.62,324.00);

\path[draw=drawColor,line width= 0.3pt,line cap=rect,fill=fillColor] (-10562.91, 34.88) rectangle (587.83, 38.55);

\path[draw=drawColor,line width= 0.3pt,line cap=rect,fill=fillColor] (-10562.91,246.93) rectangle (129.43,250.60);

\path[draw=drawColor,line width= 0.3pt,line cap=rect,fill=fillColor] (-10562.91,304.02) rectangle (115.82,307.69);

\path[draw=drawColor,line width= 0.3pt,line cap=rect,fill=fillColor] (-10562.91,140.90) rectangle (168.69,144.57);

\path[draw=drawColor,line width= 0.3pt,line cap=rect,fill=fillColor] (-10562.91,189.84) rectangle (145.17,193.51);

\path[draw=drawColor,line width= 0.3pt,line cap=rect,fill=fillColor] (-10562.91,328.48) rectangle (113.62,332.15);

\path[draw=drawColor,line width= 0.3pt,line cap=rect,fill=fillColor] (-10562.91,173.53) rectangle (147.73,177.20);

\path[draw=drawColor,line width= 0.3pt,line cap=rect,fill=fillColor] (-10562.91, 59.35) rectangle (209.99, 63.02);

\path[draw=drawColor,line width= 0.3pt,line cap=rect,fill=fillColor] (-10562.91,165.37) rectangle (152.87,169.04);

\path[draw=drawColor,line width= 0.3pt,line cap=rect,fill=fillColor] (-10562.91,149.06) rectangle (156.76,152.73);

\path[draw=drawColor,line width= 0.3pt,line cap=rect,fill=fillColor] (-10562.91,116.44) rectangle (175.10,120.11);

\path[draw=drawColor,line width= 0.3pt,line cap=rect,fill=fillColor] (-10562.91,132.75) rectangle (169.28,136.42);

\path[draw=drawColor,line width= 0.3pt,line cap=rect,fill=fillColor] (-10562.91,263.24) rectangle (128.85,266.91);

\path[draw=drawColor,line width= 0.3pt,line cap=rect,fill=fillColor] (-10562.91,206.15) rectangle (135.01,209.82);

\path[draw=drawColor,line width= 0.3pt,line cap=rect,fill=fillColor] (-10562.91, 67.50) rectangle (202.81, 71.17);

\path[draw=drawColor,line width= 0.3pt,line cap=rect,fill=fillColor] (-10562.91,287.71) rectangle (122.51,291.38);

\path[draw=drawColor,line width= 0.3pt,line cap=rect,fill=fillColor] (-10562.91,214.30) rectangle (134.45,217.97);

\path[draw=drawColor,line width= 0.3pt,line cap=rect,fill=fillColor] (-10562.91, 75.66) rectangle (199.70, 79.33);

\path[draw=drawColor,line width= 0.3pt,line cap=rect,fill=fillColor] (-10562.91,336.64) rectangle (113.62,340.31);

\path[draw=drawColor,line width= 0.3pt,line cap=rect,fill=fillColor] (-10562.91,344.79) rectangle (113.62,348.46);

\path[draw=drawColor,line width= 0.3pt,line cap=rect,fill=fillColor] (-10562.91,352.95) rectangle (113.62,356.62);

\path[draw=drawColor,line width= 0.3pt,line cap=rect,fill=fillColor] (-10562.91,238.77) rectangle (130.03,242.44);

\path[draw=drawColor,line width= 0.3pt,line cap=rect,fill=fillColor] (-10562.91, 83.81) rectangle (197.68, 87.48);

\path[draw=drawColor,line width= 0.3pt,line cap=rect,fill=fillColor] (-10562.91,100.12) rectangle (186.28,103.79);

\path[draw=drawColor,line width= 0.3pt,line cap=rect,fill=fillColor] (-10562.91,279.55) rectangle (124.86,283.22);

\path[draw=drawColor,line width= 0.3pt,line cap=rect,fill=fillColor] (-10562.91,255.08) rectangle (129.15,258.75);

\path[draw=drawColor,line width= 0.3pt,line cap=rect,fill=fillColor] (-10562.91,271.39) rectangle (128.78,275.06);

\path[draw=drawColor,line width= 0.3pt,line cap=rect,fill=fillColor] (-10562.91,124.59) rectangle (172.79,128.26);

\path[draw=drawColor,line width= 0.3pt,line cap=rect,fill=fillColor] (-10562.91,108.28) rectangle (177.36,111.95);

\path[draw=drawColor,line width= 0.3pt,line cap=rect,fill=fillColor] (-10562.91,222.46) rectangle (131.11,226.13);

\path[draw=drawColor,line width= 0.3pt,line cap=rect,fill=fillColor] (-10562.91,361.11) rectangle (113.62,364.78);

\path[draw=drawColor,line width= 0.3pt,line cap=rect,fill=fillColor] (-10562.91,312.17) rectangle (115.68,315.84);

\path[draw=drawColor,line width= 0.3pt,line cap=rect,fill=fillColor] (-10562.91,197.99) rectangle (140.04,201.66);

\path[draw=drawColor,line width= 0.3pt,line cap=rect,fill=fillColor] (-10562.91,295.86) rectangle (120.09,299.53);

\path[draw=drawColor,line width= 0.3pt,line cap=rect,fill=fillColor] (-10562.91, 91.97) rectangle (191.90, 95.64);

\path[draw=drawColor,line width= 0.3pt,line cap=rect,fill=fillColor] (-10562.91, 51.19) rectangle (211.59, 54.86);

\path[draw=drawColor,line width= 0.3pt,line cap=rect,fill=fillColor] (-10562.91,230.62) rectangle (130.68,234.29);

\path[draw=drawColor,line width= 0.3pt,line cap=rect,fill=fillColor] (-10562.91,369.26) rectangle (113.62,372.93);

\path[draw=drawColor,line width= 0.3pt,line cap=rect,fill=fillColor] (-10562.91,181.68) rectangle (147.73,185.35);

\path[draw=drawColor,line width= 0.3pt,line cap=rect,fill=fillColor] (-10562.91, 43.03) rectangle (214.89, 46.70);

\path[draw=drawColor,line width= 0.3pt,line cap=rect,fill=fillColor] (-10562.91,157.21) rectangle (156.22,160.88);
\definecolor{drawColor}{RGB}{44,83,150}

\node[text=drawColor,anchor=base,inner sep=0pt, outer sep=0pt, scale=  0.85] at (316.48, 23.78) {1.049};
\definecolor{drawColor}{RGB}{234,72,64}

\node[text=drawColor,anchor=base,inner sep=0pt, outer sep=0pt, scale=  0.85] at (316.48, 40.10) {1.044};
\definecolor{drawColor}{gray}{0.20}

\path[draw=drawColor,line width= 0.6pt,line join=round,line cap=round] (102.94, 21.83) rectangle (337.83,374.16);
\end{scope}
\begin{scope}
\path[clip] (  0.00,  0.00) rectangle (339.67,390.26);
\definecolor{drawColor}{RGB}{0,0,0}

\node[text=drawColor,anchor=base east,inner sep=0pt, outer sep=0pt, scale=  0.70] at ( 97.99, 32.47) {dblp-cite};

\node[text=drawColor,anchor=base east,inner sep=0pt, outer sep=0pt, scale=  0.70] at ( 97.99, 40.62) {petster-friendships-hamster};

\node[text=drawColor,anchor=base east,inner sep=0pt, outer sep=0pt, scale=  0.70] at ( 97.99, 48.78) {p2p-Gnutella31};

\node[text=drawColor,anchor=base east,inner sep=0pt, outer sep=0pt, scale=  0.70] at ( 97.99, 56.94) {livemocha};

\node[text=drawColor,anchor=base east,inner sep=0pt, outer sep=0pt, scale=  0.70] at ( 97.99, 65.09) {petster-friendships-dog};

\node[text=drawColor,anchor=base east,inner sep=0pt, outer sep=0pt, scale=  0.70] at ( 97.99, 73.25) {cfinder-google};

\node[text=drawColor,anchor=base east,inner sep=0pt, outer sep=0pt, scale=  0.70] at ( 97.99, 81.40) {hyves};

\node[text=drawColor,anchor=base east,inner sep=0pt, outer sep=0pt, scale=  0.70] at ( 97.99, 89.56) {bn-fly-drosophila-medulla-1};

\node[text=drawColor,anchor=base east,inner sep=0pt, outer sep=0pt, scale=  0.70] at ( 97.99, 97.71) {advogato};

\node[text=drawColor,anchor=base east,inner sep=0pt, outer sep=0pt, scale=  0.70] at ( 97.99,105.87) {bio-DM-HT};

\node[text=drawColor,anchor=base east,inner sep=0pt, outer sep=0pt, scale=  0.70] at ( 97.99,114.03) {loc-gowalla};

\node[text=drawColor,anchor=base east,inner sep=0pt, outer sep=0pt, scale=  0.70] at ( 97.99,122.18) {citeseer};

\node[text=drawColor,anchor=base east,inner sep=0pt, outer sep=0pt, scale=  0.70] at ( 97.99,130.34) {loc-brightkite};

\node[text=drawColor,anchor=base east,inner sep=0pt, outer sep=0pt, scale=  0.70] at ( 97.99,138.49) {petster-friendships-cat};

\node[text=drawColor,anchor=base east,inner sep=0pt, outer sep=0pt, scale=  0.70] at ( 97.99,146.65) {digg-friends};

\node[text=drawColor,anchor=base east,inner sep=0pt, outer sep=0pt, scale=  0.70] at ( 97.99,154.80) {petster-carnivore};

\node[text=drawColor,anchor=base east,inner sep=0pt, outer sep=0pt, scale=  0.70] at ( 97.99,162.96) {bio-CE-HT};

\node[text=drawColor,anchor=base east,inner sep=0pt, outer sep=0pt, scale=  0.70] at ( 97.99,171.11) {moreno-propro};

\node[text=drawColor,anchor=base east,inner sep=0pt, outer sep=0pt, scale=  0.70] at ( 97.99,179.27) {bio-yeast-protein-inter};

\node[text=drawColor,anchor=base east,inner sep=0pt, outer sep=0pt, scale=  0.70] at ( 97.99,187.43) {web-Google};

\node[text=drawColor,anchor=base east,inner sep=0pt, outer sep=0pt, scale=  0.70] at ( 97.99,195.58) {as-skitter};

\node[text=drawColor,anchor=base east,inner sep=0pt, outer sep=0pt, scale=  0.70] at ( 97.99,203.74) {com-amazon};

\node[text=drawColor,anchor=base east,inner sep=0pt, outer sep=0pt, scale=  0.70] at ( 97.99,211.89) {moreno-names};

\node[text=drawColor,anchor=base east,inner sep=0pt, outer sep=0pt, scale=  0.70] at ( 97.99,220.05) {ca-cit-HepPh};

\node[text=drawColor,anchor=base east,inner sep=0pt, outer sep=0pt, scale=  0.70] at ( 97.99,228.20) {ca-cit-HepTh};

\node[text=drawColor,anchor=base east,inner sep=0pt, outer sep=0pt, scale=  0.70] at ( 97.99,236.36) {flixster};

\node[text=drawColor,anchor=base east,inner sep=0pt, outer sep=0pt, scale=  0.70] at ( 97.99,244.52) {soc-Epinions1};

\node[text=drawColor,anchor=base east,inner sep=0pt, outer sep=0pt, scale=  0.70] at ( 97.99,252.67) {youtube-u-growth};

\node[text=drawColor,anchor=base east,inner sep=0pt, outer sep=0pt, scale=  0.70] at ( 97.99,260.83) {com-youtube};

\node[text=drawColor,anchor=base east,inner sep=0pt, outer sep=0pt, scale=  0.70] at ( 97.99,268.98) {youtube-links};

\node[text=drawColor,anchor=base east,inner sep=0pt, outer sep=0pt, scale=  0.70] at ( 97.99,277.14) {wordnet-words};

\node[text=drawColor,anchor=base east,inner sep=0pt, outer sep=0pt, scale=  0.70] at ( 97.99,285.29) {ca-AstroPh};

\node[text=drawColor,anchor=base east,inner sep=0pt, outer sep=0pt, scale=  0.70] at ( 97.99,293.45) {as-22july06};

\node[text=drawColor,anchor=base east,inner sep=0pt, outer sep=0pt, scale=  0.70] at ( 97.99,301.61) {com-dblp};

\node[text=drawColor,anchor=base east,inner sep=0pt, outer sep=0pt, scale=  0.70] at ( 97.99,309.76) {topology};

\node[text=drawColor,anchor=base east,inner sep=0pt, outer sep=0pt, scale=  0.70] at ( 97.99,317.92) {ego-facebook};

\node[text=drawColor,anchor=base east,inner sep=0pt, outer sep=0pt, scale=  0.70] at ( 97.99,326.07) {bn-mouse-kasthuri-graph-v4};

\node[text=drawColor,anchor=base east,inner sep=0pt, outer sep=0pt, scale=  0.70] at ( 97.99,334.23) {ego-gplus};

\node[text=drawColor,anchor=base east,inner sep=0pt, outer sep=0pt, scale=  0.70] at ( 97.99,342.38) {as-caida20071105};

\node[text=drawColor,anchor=base east,inner sep=0pt, outer sep=0pt, scale=  0.70] at ( 97.99,350.54) {bio-CE-LC};

\node[text=drawColor,anchor=base east,inner sep=0pt, outer sep=0pt, scale=  0.70] at ( 97.99,358.70) {munmun-twitter-social};

\node[text=drawColor,anchor=base east,inner sep=0pt, outer sep=0pt, scale=  0.70] at ( 97.99,366.85) {as20000102};
\end{scope}
\begin{scope}
\path[clip] (  0.00,  0.00) rectangle (339.67,390.26);
\definecolor{drawColor}{gray}{0.20}

\path[draw=drawColor,line width= 0.6pt,line join=round] (100.19, 34.88) --
	(102.94, 34.88);

\path[draw=drawColor,line width= 0.6pt,line join=round] (100.19, 43.03) --
	(102.94, 43.03);

\path[draw=drawColor,line width= 0.6pt,line join=round] (100.19, 51.19) --
	(102.94, 51.19);

\path[draw=drawColor,line width= 0.6pt,line join=round] (100.19, 59.35) --
	(102.94, 59.35);

\path[draw=drawColor,line width= 0.6pt,line join=round] (100.19, 67.50) --
	(102.94, 67.50);

\path[draw=drawColor,line width= 0.6pt,line join=round] (100.19, 75.66) --
	(102.94, 75.66);

\path[draw=drawColor,line width= 0.6pt,line join=round] (100.19, 83.81) --
	(102.94, 83.81);

\path[draw=drawColor,line width= 0.6pt,line join=round] (100.19, 91.97) --
	(102.94, 91.97);

\path[draw=drawColor,line width= 0.6pt,line join=round] (100.19,100.12) --
	(102.94,100.12);

\path[draw=drawColor,line width= 0.6pt,line join=round] (100.19,108.28) --
	(102.94,108.28);

\path[draw=drawColor,line width= 0.6pt,line join=round] (100.19,116.44) --
	(102.94,116.44);

\path[draw=drawColor,line width= 0.6pt,line join=round] (100.19,124.59) --
	(102.94,124.59);

\path[draw=drawColor,line width= 0.6pt,line join=round] (100.19,132.75) --
	(102.94,132.75);

\path[draw=drawColor,line width= 0.6pt,line join=round] (100.19,140.90) --
	(102.94,140.90);

\path[draw=drawColor,line width= 0.6pt,line join=round] (100.19,149.06) --
	(102.94,149.06);

\path[draw=drawColor,line width= 0.6pt,line join=round] (100.19,157.21) --
	(102.94,157.21);

\path[draw=drawColor,line width= 0.6pt,line join=round] (100.19,165.37) --
	(102.94,165.37);

\path[draw=drawColor,line width= 0.6pt,line join=round] (100.19,173.53) --
	(102.94,173.53);

\path[draw=drawColor,line width= 0.6pt,line join=round] (100.19,181.68) --
	(102.94,181.68);

\path[draw=drawColor,line width= 0.6pt,line join=round] (100.19,189.84) --
	(102.94,189.84);

\path[draw=drawColor,line width= 0.6pt,line join=round] (100.19,197.99) --
	(102.94,197.99);

\path[draw=drawColor,line width= 0.6pt,line join=round] (100.19,206.15) --
	(102.94,206.15);

\path[draw=drawColor,line width= 0.6pt,line join=round] (100.19,214.30) --
	(102.94,214.30);

\path[draw=drawColor,line width= 0.6pt,line join=round] (100.19,222.46) --
	(102.94,222.46);

\path[draw=drawColor,line width= 0.6pt,line join=round] (100.19,230.62) --
	(102.94,230.62);

\path[draw=drawColor,line width= 0.6pt,line join=round] (100.19,238.77) --
	(102.94,238.77);

\path[draw=drawColor,line width= 0.6pt,line join=round] (100.19,246.93) --
	(102.94,246.93);

\path[draw=drawColor,line width= 0.6pt,line join=round] (100.19,255.08) --
	(102.94,255.08);

\path[draw=drawColor,line width= 0.6pt,line join=round] (100.19,263.24) --
	(102.94,263.24);

\path[draw=drawColor,line width= 0.6pt,line join=round] (100.19,271.39) --
	(102.94,271.39);

\path[draw=drawColor,line width= 0.6pt,line join=round] (100.19,279.55) --
	(102.94,279.55);

\path[draw=drawColor,line width= 0.6pt,line join=round] (100.19,287.71) --
	(102.94,287.71);

\path[draw=drawColor,line width= 0.6pt,line join=round] (100.19,295.86) --
	(102.94,295.86);

\path[draw=drawColor,line width= 0.6pt,line join=round] (100.19,304.02) --
	(102.94,304.02);

\path[draw=drawColor,line width= 0.6pt,line join=round] (100.19,312.17) --
	(102.94,312.17);

\path[draw=drawColor,line width= 0.6pt,line join=round] (100.19,320.33) --
	(102.94,320.33);

\path[draw=drawColor,line width= 0.6pt,line join=round] (100.19,328.48) --
	(102.94,328.48);

\path[draw=drawColor,line width= 0.6pt,line join=round] (100.19,336.64) --
	(102.94,336.64);

\path[draw=drawColor,line width= 0.6pt,line join=round] (100.19,344.79) --
	(102.94,344.79);

\path[draw=drawColor,line width= 0.6pt,line join=round] (100.19,352.95) --
	(102.94,352.95);

\path[draw=drawColor,line width= 0.6pt,line join=round] (100.19,361.11) --
	(102.94,361.11);

\path[draw=drawColor,line width= 0.6pt,line join=round] (100.19,369.26) --
	(102.94,369.26);
\end{scope}
\begin{scope}
\path[clip] (  0.00,  0.00) rectangle (339.67,390.26);
\definecolor{drawColor}{gray}{0.20}

\path[draw=drawColor,line width= 0.6pt,line join=round] (113.62, 19.08) --
	(113.62, 21.83);

\path[draw=drawColor,line width= 0.6pt,line join=round] (167.00, 19.08) --
	(167.00, 21.83);

\path[draw=drawColor,line width= 0.6pt,line join=round] (220.39, 19.08) --
	(220.39, 21.83);

\path[draw=drawColor,line width= 0.6pt,line join=round] (273.77, 19.08) --
	(273.77, 21.83);

\path[draw=drawColor,line width= 0.6pt,line join=round] (327.15, 19.08) --
	(327.15, 21.83);
\end{scope}
\begin{scope}
\path[clip] (  0.00,  0.00) rectangle (339.67,390.26);
\definecolor{drawColor}{RGB}{0,0,0}

\node[text=drawColor,anchor=base,inner sep=0pt, outer sep=0pt, scale=  0.70] at (113.62, 12.06) {1.000};

\node[text=drawColor,anchor=base,inner sep=0pt, outer sep=0pt, scale=  0.70] at (167.00, 12.06) {1.005};

\node[text=drawColor,anchor=base,inner sep=0pt, outer sep=0pt, scale=  0.70] at (220.39, 12.06) {1.010};

\node[text=drawColor,anchor=base,inner sep=0pt, outer sep=0pt, scale=  0.70] at (273.77, 12.06) {1.015};

\node[text=drawColor,anchor=base,inner sep=0pt, outer sep=0pt, scale=  0.70] at (327.15, 12.06) {1.020};
\end{scope}
\begin{scope}
\path[clip] (  0.00,  0.00) rectangle (339.67,390.26);
\definecolor{drawColor}{RGB}{0,0,0}

\node[text=drawColor,anchor=base,inner sep=0pt, outer sep=0pt, scale=  0.90] at (220.39,  1.75) {Approximation Ratio};
\end{scope}
\begin{scope}
\path[clip] (  0.00,  0.00) rectangle (339.67,390.26);
\definecolor{drawColor}{RGB}{0,0,0}

\node[text=drawColor,rotate= 90.00,anchor=base,inner sep=0pt, outer sep=0pt, scale=  0.90] at (  8.04,197.99) {Network};
\end{scope}
\begin{scope}
\path[clip] (  0.00,  0.00) rectangle (339.67,390.26);
\definecolor{fillColor}{RGB}{255,255,255}

\path[fill=fillColor] (107.88,380.16) rectangle (332.89,398.10);
\end{scope}
\begin{scope}
\path[clip] (  0.00,  0.00) rectangle (339.67,390.26);
\definecolor{drawColor}{RGB}{0,0,0}

\node[text=drawColor,anchor=base west,inner sep=0pt, outer sep=0pt, scale=  0.90] at (107.88,381.03) {Approximation Algorithm};
\end{scope}
\begin{scope}
\path[clip] (  0.00,  0.00) rectangle (339.67,390.26);
\definecolor{fillColor}{RGB}{255,255,255}

\path[fill=fillColor] (219.91,380.99) rectangle (241.25,388.10);
\end{scope}
\begin{scope}
\path[clip] (  0.00,  0.00) rectangle (339.67,390.26);
\definecolor{drawColor}{RGB}{0,0,0}
\definecolor{fillColor}{RGB}{44,83,150}

\path[draw=drawColor,line width= 0.3pt,line cap=rect,fill=fillColor] (220.26,381.35) rectangle (240.89,387.75);
\end{scope}
\begin{scope}
\path[clip] (  0.00,  0.00) rectangle (339.67,390.26);
\definecolor{fillColor}{RGB}{255,255,255}

\path[fill=fillColor] (278.26,380.99) rectangle (299.60,388.10);
\end{scope}
\begin{scope}
\path[clip] (  0.00,  0.00) rectangle (339.67,390.26);
\definecolor{drawColor}{RGB}{0,0,0}
\definecolor{fillColor}{RGB}{234,72,64}

\path[draw=drawColor,line width= 0.3pt,line cap=rect,fill=fillColor] (278.62,381.35) rectangle (299.24,387.75);
\end{scope}
\begin{scope}
\path[clip] (  0.00,  0.00) rectangle (339.67,390.26);
\definecolor{drawColor}{RGB}{0,0,0}

\node[text=drawColor,anchor=base west,inner sep=0pt, outer sep=0pt, scale=  0.70] at (245.75,382.14) {Standard};
\end{scope}
\begin{scope}
\path[clip] (  0.00,  0.00) rectangle (339.67,390.26);
\definecolor{drawColor}{RGB}{0,0,0}

\node[text=drawColor,anchor=base west,inner sep=0pt, outer sep=0pt, scale=  0.70] at (304.10,382.14) {Improved};
\end{scope}
\end{tikzpicture}

%% file: plot_relative.tex
\begin{tikzpicture}[x=1pt,y=1pt]
\definecolor{fillColor}{RGB}{255,255,255}
\path[use as bounding box,fill=fillColor,fill opacity=0.00] (0,0) rectangle (339.67,375.80);
\begin{scope}
\path[clip] (  0.00,  0.00) rectangle (339.67,375.80);
\definecolor{drawColor}{RGB}{255,255,255}
\definecolor{fillColor}{RGB}{255,255,255}

\path[draw=drawColor,line width= 0.6pt,line join=round,line cap=round,fill=fillColor] ( -0.00,  0.00) rectangle (339.67,375.80);
\end{scope}
\begin{scope}
\path[clip] (106.60, 27.33) rectangle (334.17,370.30);
\definecolor{fillColor}{RGB}{255,255,255}

\path[fill=fillColor] (106.60, 27.33) rectangle (334.17,370.30);
\definecolor{drawColor}{gray}{0.92}

\path[draw=drawColor,line width= 0.3pt,line join=round] (142.81, 27.33) --
	(142.81,370.30);

\path[draw=drawColor,line width= 0.3pt,line join=round] (194.53, 27.33) --
	(194.53,370.30);

\path[draw=drawColor,line width= 0.3pt,line join=round] (246.25, 27.33) --
	(246.25,370.30);

\path[draw=drawColor,line width= 0.3pt,line join=round] (297.97, 27.33) --
	(297.97,370.30);

\path[draw=drawColor,line width= 0.6pt,line join=round] (106.60, 32.21) --
	(334.17, 32.21);

\path[draw=drawColor,line width= 0.6pt,line join=round] (106.60, 40.33) --
	(334.17, 40.33);

\path[draw=drawColor,line width= 0.6pt,line join=round] (106.60, 48.46) --
	(334.17, 48.46);

\path[draw=drawColor,line width= 0.6pt,line join=round] (106.60, 56.59) --
	(334.17, 56.59);

\path[draw=drawColor,line width= 0.6pt,line join=round] (106.60, 64.72) --
	(334.17, 64.72);

\path[draw=drawColor,line width= 0.6pt,line join=round] (106.60, 72.84) --
	(334.17, 72.84);

\path[draw=drawColor,line width= 0.6pt,line join=round] (106.60, 80.97) --
	(334.17, 80.97);

\path[draw=drawColor,line width= 0.6pt,line join=round] (106.60, 89.10) --
	(334.17, 89.10);

\path[draw=drawColor,line width= 0.6pt,line join=round] (106.60, 97.22) --
	(334.17, 97.22);

\path[draw=drawColor,line width= 0.6pt,line join=round] (106.60,105.35) --
	(334.17,105.35);

\path[draw=drawColor,line width= 0.6pt,line join=round] (106.60,113.48) --
	(334.17,113.48);

\path[draw=drawColor,line width= 0.6pt,line join=round] (106.60,121.61) --
	(334.17,121.61);

\path[draw=drawColor,line width= 0.6pt,line join=round] (106.60,129.73) --
	(334.17,129.73);

\path[draw=drawColor,line width= 0.6pt,line join=round] (106.60,137.86) --
	(334.17,137.86);

\path[draw=drawColor,line width= 0.6pt,line join=round] (106.60,145.99) --
	(334.17,145.99);

\path[draw=drawColor,line width= 0.6pt,line join=round] (106.60,154.12) --
	(334.17,154.12);

\path[draw=drawColor,line width= 0.6pt,line join=round] (106.60,162.24) --
	(334.17,162.24);

\path[draw=drawColor,line width= 0.6pt,line join=round] (106.60,170.37) --
	(334.17,170.37);

\path[draw=drawColor,line width= 0.6pt,line join=round] (106.60,178.50) --
	(334.17,178.50);

\path[draw=drawColor,line width= 0.6pt,line join=round] (106.60,186.63) --
	(334.17,186.63);

\path[draw=drawColor,line width= 0.6pt,line join=round] (106.60,194.75) --
	(334.17,194.75);

\path[draw=drawColor,line width= 0.6pt,line join=round] (106.60,202.88) --
	(334.17,202.88);

\path[draw=drawColor,line width= 0.6pt,line join=round] (106.60,211.01) --
	(334.17,211.01);

\path[draw=drawColor,line width= 0.6pt,line join=round] (106.60,219.14) --
	(334.17,219.14);

\path[draw=drawColor,line width= 0.6pt,line join=round] (106.60,227.26) --
	(334.17,227.26);

\path[draw=drawColor,line width= 0.6pt,line join=round] (106.60,235.39) --
	(334.17,235.39);

\path[draw=drawColor,line width= 0.6pt,line join=round] (106.60,243.52) --
	(334.17,243.52);

\path[draw=drawColor,line width= 0.6pt,line join=round] (106.60,251.64) --
	(334.17,251.64);

\path[draw=drawColor,line width= 0.6pt,line join=round] (106.60,259.77) --
	(334.17,259.77);

\path[draw=drawColor,line width= 0.6pt,line join=round] (106.60,267.90) --
	(334.17,267.90);

\path[draw=drawColor,line width= 0.6pt,line join=round] (106.60,276.03) --
	(334.17,276.03);

\path[draw=drawColor,line width= 0.6pt,line join=round] (106.60,284.15) --
	(334.17,284.15);

\path[draw=drawColor,line width= 0.6pt,line join=round] (106.60,292.28) --
	(334.17,292.28);

\path[draw=drawColor,line width= 0.6pt,line join=round] (106.60,300.41) --
	(334.17,300.41);

\path[draw=drawColor,line width= 0.6pt,line join=round] (106.60,308.54) --
	(334.17,308.54);

\path[draw=drawColor,line width= 0.6pt,line join=round] (106.60,316.66) --
	(334.17,316.66);

\path[draw=drawColor,line width= 0.6pt,line join=round] (106.60,324.79) --
	(334.17,324.79);

\path[draw=drawColor,line width= 0.6pt,line join=round] (106.60,332.92) --
	(334.17,332.92);

\path[draw=drawColor,line width= 0.6pt,line join=round] (106.60,341.05) --
	(334.17,341.05);

\path[draw=drawColor,line width= 0.6pt,line join=round] (106.60,349.17) --
	(334.17,349.17);

\path[draw=drawColor,line width= 0.6pt,line join=round] (106.60,357.30) --
	(334.17,357.30);

\path[draw=drawColor,line width= 0.6pt,line join=round] (106.60,365.43) --
	(334.17,365.43);

\path[draw=drawColor,line width= 0.6pt,line join=round] (116.95, 27.33) --
	(116.95,370.30);

\path[draw=drawColor,line width= 0.6pt,line join=round] (168.67, 27.33) --
	(168.67,370.30);

\path[draw=drawColor,line width= 0.6pt,line join=round] (220.39, 27.33) --
	(220.39,370.30);

\path[draw=drawColor,line width= 0.6pt,line join=round] (272.11, 27.33) --
	(272.11,370.30);

\path[draw=drawColor,line width= 0.6pt,line join=round] (323.83, 27.33) --
	(323.83,370.30);
\definecolor{drawColor}{RGB}{0,0,0}
\definecolor{fillColor}{gray}{0.60}

\path[draw=drawColor,line width= 0.3pt,line cap=rect,fill=fillColor] (116.95, 29.77) rectangle (323.83, 34.64);
\definecolor{fillColor}{RGB}{0,158,115}

\path[draw=drawColor,line width= 0.3pt,line cap=rect,fill=fillColor] (116.95, 78.53) rectangle (306.01, 83.41);

\path[draw=drawColor,line width= 0.3pt,line cap=rect,fill=fillColor] (116.95,273.59) rectangle (165.71,278.46);

\path[draw=drawColor,line width= 0.3pt,line cap=rect,fill=fillColor] (116.95,322.35) rectangle (140.63,327.23);

\path[draw=drawColor,line width= 0.3pt,line cap=rect,fill=fillColor] (116.95,111.04) rectangle (232.06,115.92);

\path[draw=drawColor,line width= 0.3pt,line cap=rect,fill=fillColor] (116.95,224.82) rectangle (188.52,229.70);

\path[draw=drawColor,line width= 0.3pt,line cap=rect,fill=fillColor] (116.95,362.99) rectangle (116.95,367.87);

\path[draw=drawColor,line width= 0.3pt,line cap=rect,fill=fillColor] (116.95,192.32) rectangle (199.70,197.19);

\path[draw=drawColor,line width= 0.3pt,line cap=rect,fill=fillColor] (116.95,143.55) rectangle (218.19,148.43);

\path[draw=drawColor,line width= 0.3pt,line cap=rect,fill=fillColor] (116.95,281.72) rectangle (164.69,286.59);

\path[draw=drawColor,line width= 0.3pt,line cap=rect,fill=fillColor] (116.95,151.68) rectangle (213.55,156.55);

\path[draw=drawColor,line width= 0.3pt,line cap=rect,fill=fillColor] (116.95,184.19) rectangle (200.01,189.06);

\path[draw=drawColor,line width= 0.3pt,line cap=rect,fill=fillColor] (116.95,208.57) rectangle (191.35,213.45);

\path[draw=drawColor,line width= 0.3pt,line cap=rect,fill=fillColor] (116.95,306.10) rectangle (152.93,310.97);

\path[draw=drawColor,line width= 0.3pt,line cap=rect,fill=fillColor] (116.95,289.84) rectangle (155.51,294.72);

\path[draw=drawColor,line width= 0.3pt,line cap=rect,fill=fillColor] (116.95,102.91) rectangle (236.00,107.79);

\path[draw=drawColor,line width= 0.3pt,line cap=rect,fill=fillColor] (116.95,241.08) rectangle (177.79,245.96);

\path[draw=drawColor,line width= 0.3pt,line cap=rect,fill=fillColor] (116.95,249.21) rectangle (176.06,254.08);

\path[draw=drawColor,line width= 0.3pt,line cap=rect,fill=fillColor] (116.95, 62.28) rectangle (317.56, 67.15);
\definecolor{fillColor}{gray}{0.60}

\path[draw=drawColor,line width= 0.3pt,line cap=rect,fill=fillColor] (116.95, 37.90) rectangle (323.83, 42.77);
\definecolor{fillColor}{RGB}{0,158,115}

\path[draw=drawColor,line width= 0.3pt,line cap=rect,fill=fillColor] (116.95,338.61) rectangle (116.95,343.48);

\path[draw=drawColor,line width= 0.3pt,line cap=rect,fill=fillColor] (116.95,354.86) rectangle (116.95,359.74);

\path[draw=drawColor,line width= 0.3pt,line cap=rect,fill=fillColor] (116.95, 86.66) rectangle (266.30, 91.54);

\path[draw=drawColor,line width= 0.3pt,line cap=rect,fill=fillColor] (116.95, 54.15) rectangle (322.57, 59.03);

\path[draw=drawColor,line width= 0.3pt,line cap=rect,fill=fillColor] (116.95,127.30) rectangle (223.95,132.17);

\path[draw=drawColor,line width= 0.3pt,line cap=rect,fill=fillColor] (116.95,265.46) rectangle (169.72,270.34);

\path[draw=drawColor,line width= 0.3pt,line cap=rect,fill=fillColor] (116.95,314.23) rectangle (152.19,319.10);

\path[draw=drawColor,line width= 0.3pt,line cap=rect,fill=fillColor] (116.95,297.97) rectangle (152.95,302.85);

\path[draw=drawColor,line width= 0.3pt,line cap=rect,fill=fillColor] (116.95,200.44) rectangle (193.00,205.32);

\path[draw=drawColor,line width= 0.3pt,line cap=rect,fill=fillColor] (116.95,216.70) rectangle (189.96,221.57);

\path[draw=drawColor,line width= 0.3pt,line cap=rect,fill=fillColor] (116.95,176.06) rectangle (205.95,180.94);
\definecolor{fillColor}{gray}{0.60}

\path[draw=drawColor,line width= 0.3pt,line cap=rect,fill=fillColor] (116.95, 46.02) rectangle (323.83, 50.90);
\definecolor{fillColor}{RGB}{0,158,115}

\path[draw=drawColor,line width= 0.3pt,line cap=rect,fill=fillColor] (116.95,330.48) rectangle (129.88,335.36);

\path[draw=drawColor,line width= 0.3pt,line cap=rect,fill=fillColor] (116.95,257.33) rectangle (170.55,262.21);

\path[draw=drawColor,line width= 0.3pt,line cap=rect,fill=fillColor] (116.95,232.95) rectangle (185.91,237.83);

\path[draw=drawColor,line width= 0.3pt,line cap=rect,fill=fillColor] (116.95,159.81) rectangle (210.98,164.68);

\path[draw=drawColor,line width= 0.3pt,line cap=rect,fill=fillColor] (116.95, 70.40) rectangle (311.66, 75.28);

\path[draw=drawColor,line width= 0.3pt,line cap=rect,fill=fillColor] (116.95,167.93) rectangle (210.69,172.81);

\path[draw=drawColor,line width= 0.3pt,line cap=rect,fill=fillColor] (116.95,346.73) rectangle (116.95,351.61);

\path[draw=drawColor,line width= 0.3pt,line cap=rect,fill=fillColor] (116.95,135.42) rectangle (220.39,140.30);

\path[draw=drawColor,line width= 0.3pt,line cap=rect,fill=fillColor] (116.95, 94.79) rectangle (261.76, 99.66);

\path[draw=drawColor,line width= 0.3pt,line cap=rect,fill=fillColor] (116.95,119.17) rectangle (226.92,124.05);
\definecolor{drawColor}{gray}{0.20}

\path[draw=drawColor,line width= 0.6pt,line join=round,line cap=round] (106.60, 27.33) rectangle (334.17,370.30);
\end{scope}
\begin{scope}
\path[clip] (  0.00,  0.00) rectangle (339.67,375.80);
\definecolor{drawColor}{RGB}{0,0,0}

\node[text=drawColor,anchor=base east,inner sep=0pt, outer sep=0pt, scale=  0.70] at (101.65, 29.80) {ego-facebook};

\node[text=drawColor,anchor=base east,inner sep=0pt, outer sep=0pt, scale=  0.70] at (101.65, 37.92) {ego-gplus};

\node[text=drawColor,anchor=base east,inner sep=0pt, outer sep=0pt, scale=  0.70] at (101.65, 46.05) {munmun-twitter-social};

\node[text=drawColor,anchor=base east,inner sep=0pt, outer sep=0pt, scale=  0.70] at (101.65, 54.18) {hyves};

\node[text=drawColor,anchor=base east,inner sep=0pt, outer sep=0pt, scale=  0.70] at (101.65, 62.31) {cfinder-google};

\node[text=drawColor,anchor=base east,inner sep=0pt, outer sep=0pt, scale=  0.70] at (101.65, 70.43) {p2p-Gnutella31};

\node[text=drawColor,anchor=base east,inner sep=0pt, outer sep=0pt, scale=  0.70] at (101.65, 78.56) {dblp-cite};

\node[text=drawColor,anchor=base east,inner sep=0pt, outer sep=0pt, scale=  0.70] at (101.65, 86.69) {flixster};

\node[text=drawColor,anchor=base east,inner sep=0pt, outer sep=0pt, scale=  0.70] at (101.65, 94.81) {petster-friendships-hamster};

\node[text=drawColor,anchor=base east,inner sep=0pt, outer sep=0pt, scale=  0.70] at (101.65,102.94) {petster-friendships-dog};

\node[text=drawColor,anchor=base east,inner sep=0pt, outer sep=0pt, scale=  0.70] at (101.65,111.07) {petster-friendships-cat};

\node[text=drawColor,anchor=base east,inner sep=0pt, outer sep=0pt, scale=  0.70] at (101.65,119.20) {petster-carnivore};

\node[text=drawColor,anchor=base east,inner sep=0pt, outer sep=0pt, scale=  0.70] at (101.65,127.32) {advogato};

\node[text=drawColor,anchor=base east,inner sep=0pt, outer sep=0pt, scale=  0.70] at (101.65,135.45) {bio-yeast-protein-inter};

\node[text=drawColor,anchor=base east,inner sep=0pt, outer sep=0pt, scale=  0.70] at (101.65,143.58) {livemocha};

\node[text=drawColor,anchor=base east,inner sep=0pt, outer sep=0pt, scale=  0.70] at (101.65,151.71) {digg-friends};

\node[text=drawColor,anchor=base east,inner sep=0pt, outer sep=0pt, scale=  0.70] at (101.65,159.83) {bn-fly-drosophila-medulla-1};

\node[text=drawColor,anchor=base east,inner sep=0pt, outer sep=0pt, scale=  0.70] at (101.65,167.96) {ca-cit-HepTh};

\node[text=drawColor,anchor=base east,inner sep=0pt, outer sep=0pt, scale=  0.70] at (101.65,176.09) {ca-cit-HepPh};

\node[text=drawColor,anchor=base east,inner sep=0pt, outer sep=0pt, scale=  0.70] at (101.65,184.22) {loc-gowalla};

\node[text=drawColor,anchor=base east,inner sep=0pt, outer sep=0pt, scale=  0.70] at (101.65,192.34) {moreno-propro};

\node[text=drawColor,anchor=base east,inner sep=0pt, outer sep=0pt, scale=  0.70] at (101.65,200.47) {citeseer};

\node[text=drawColor,anchor=base east,inner sep=0pt, outer sep=0pt, scale=  0.70] at (101.65,208.60) {loc-brightkite};

\node[text=drawColor,anchor=base east,inner sep=0pt, outer sep=0pt, scale=  0.70] at (101.65,216.72) {bio-DM-HT};

\node[text=drawColor,anchor=base east,inner sep=0pt, outer sep=0pt, scale=  0.70] at (101.65,224.85) {web-Google};

\node[text=drawColor,anchor=base east,inner sep=0pt, outer sep=0pt, scale=  0.70] at (101.65,232.98) {as-22july06};

\node[text=drawColor,anchor=base east,inner sep=0pt, outer sep=0pt, scale=  0.70] at (101.65,241.11) {ca-AstroPh};

\node[text=drawColor,anchor=base east,inner sep=0pt, outer sep=0pt, scale=  0.70] at (101.65,249.23) {moreno-names};

\node[text=drawColor,anchor=base east,inner sep=0pt, outer sep=0pt, scale=  0.70] at (101.65,257.36) {as-skitter};

\node[text=drawColor,anchor=base east,inner sep=0pt, outer sep=0pt, scale=  0.70] at (101.65,265.49) {wordnet-words};

\node[text=drawColor,anchor=base east,inner sep=0pt, outer sep=0pt, scale=  0.70] at (101.65,273.62) {soc-Epinions1};

\node[text=drawColor,anchor=base east,inner sep=0pt, outer sep=0pt, scale=  0.70] at (101.65,281.74) {bio-CE-HT};

\node[text=drawColor,anchor=base east,inner sep=0pt, outer sep=0pt, scale=  0.70] at (101.65,289.87) {com-amazon};

\node[text=drawColor,anchor=base east,inner sep=0pt, outer sep=0pt, scale=  0.70] at (101.65,298.00) {youtube-links};

\node[text=drawColor,anchor=base east,inner sep=0pt, outer sep=0pt, scale=  0.70] at (101.65,306.13) {com-youtube};

\node[text=drawColor,anchor=base east,inner sep=0pt, outer sep=0pt, scale=  0.70] at (101.65,314.25) {youtube-u-growth};

\node[text=drawColor,anchor=base east,inner sep=0pt, outer sep=0pt, scale=  0.70] at (101.65,322.38) {com-dblp};

\node[text=drawColor,anchor=base east,inner sep=0pt, outer sep=0pt, scale=  0.70] at (101.65,330.51) {topology};

\node[text=drawColor,anchor=base east,inner sep=0pt, outer sep=0pt, scale=  0.70] at (101.65,338.64) {as-caida20071105};

\node[text=drawColor,anchor=base east,inner sep=0pt, outer sep=0pt, scale=  0.70] at (101.65,346.76) {as20000102};

\node[text=drawColor,anchor=base east,inner sep=0pt, outer sep=0pt, scale=  0.70] at (101.65,354.89) {bio-CE-LC};

\node[text=drawColor,anchor=base east,inner sep=0pt, outer sep=0pt, scale=  0.70] at (101.65,363.02) {bn-mouse-kasthuri-graph-v4};
\end{scope}
\begin{scope}
\path[clip] (  0.00,  0.00) rectangle (339.67,375.80);
\definecolor{drawColor}{gray}{0.20}

\path[draw=drawColor,line width= 0.6pt,line join=round] (103.85, 32.21) --
	(106.60, 32.21);

\path[draw=drawColor,line width= 0.6pt,line join=round] (103.85, 40.33) --
	(106.60, 40.33);

\path[draw=drawColor,line width= 0.6pt,line join=round] (103.85, 48.46) --
	(106.60, 48.46);

\path[draw=drawColor,line width= 0.6pt,line join=round] (103.85, 56.59) --
	(106.60, 56.59);

\path[draw=drawColor,line width= 0.6pt,line join=round] (103.85, 64.72) --
	(106.60, 64.72);

\path[draw=drawColor,line width= 0.6pt,line join=round] (103.85, 72.84) --
	(106.60, 72.84);

\path[draw=drawColor,line width= 0.6pt,line join=round] (103.85, 80.97) --
	(106.60, 80.97);

\path[draw=drawColor,line width= 0.6pt,line join=round] (103.85, 89.10) --
	(106.60, 89.10);

\path[draw=drawColor,line width= 0.6pt,line join=round] (103.85, 97.22) --
	(106.60, 97.22);

\path[draw=drawColor,line width= 0.6pt,line join=round] (103.85,105.35) --
	(106.60,105.35);

\path[draw=drawColor,line width= 0.6pt,line join=round] (103.85,113.48) --
	(106.60,113.48);

\path[draw=drawColor,line width= 0.6pt,line join=round] (103.85,121.61) --
	(106.60,121.61);

\path[draw=drawColor,line width= 0.6pt,line join=round] (103.85,129.73) --
	(106.60,129.73);

\path[draw=drawColor,line width= 0.6pt,line join=round] (103.85,137.86) --
	(106.60,137.86);

\path[draw=drawColor,line width= 0.6pt,line join=round] (103.85,145.99) --
	(106.60,145.99);

\path[draw=drawColor,line width= 0.6pt,line join=round] (103.85,154.12) --
	(106.60,154.12);

\path[draw=drawColor,line width= 0.6pt,line join=round] (103.85,162.24) --
	(106.60,162.24);

\path[draw=drawColor,line width= 0.6pt,line join=round] (103.85,170.37) --
	(106.60,170.37);

\path[draw=drawColor,line width= 0.6pt,line join=round] (103.85,178.50) --
	(106.60,178.50);

\path[draw=drawColor,line width= 0.6pt,line join=round] (103.85,186.63) --
	(106.60,186.63);

\path[draw=drawColor,line width= 0.6pt,line join=round] (103.85,194.75) --
	(106.60,194.75);

\path[draw=drawColor,line width= 0.6pt,line join=round] (103.85,202.88) --
	(106.60,202.88);

\path[draw=drawColor,line width= 0.6pt,line join=round] (103.85,211.01) --
	(106.60,211.01);

\path[draw=drawColor,line width= 0.6pt,line join=round] (103.85,219.14) --
	(106.60,219.14);

\path[draw=drawColor,line width= 0.6pt,line join=round] (103.85,227.26) --
	(106.60,227.26);

\path[draw=drawColor,line width= 0.6pt,line join=round] (103.85,235.39) --
	(106.60,235.39);

\path[draw=drawColor,line width= 0.6pt,line join=round] (103.85,243.52) --
	(106.60,243.52);

\path[draw=drawColor,line width= 0.6pt,line join=round] (103.85,251.64) --
	(106.60,251.64);

\path[draw=drawColor,line width= 0.6pt,line join=round] (103.85,259.77) --
	(106.60,259.77);

\path[draw=drawColor,line width= 0.6pt,line join=round] (103.85,267.90) --
	(106.60,267.90);

\path[draw=drawColor,line width= 0.6pt,line join=round] (103.85,276.03) --
	(106.60,276.03);

\path[draw=drawColor,line width= 0.6pt,line join=round] (103.85,284.15) --
	(106.60,284.15);

\path[draw=drawColor,line width= 0.6pt,line join=round] (103.85,292.28) --
	(106.60,292.28);

\path[draw=drawColor,line width= 0.6pt,line join=round] (103.85,300.41) --
	(106.60,300.41);

\path[draw=drawColor,line width= 0.6pt,line join=round] (103.85,308.54) --
	(106.60,308.54);

\path[draw=drawColor,line width= 0.6pt,line join=round] (103.85,316.66) --
	(106.60,316.66);

\path[draw=drawColor,line width= 0.6pt,line join=round] (103.85,324.79) --
	(106.60,324.79);

\path[draw=drawColor,line width= 0.6pt,line join=round] (103.85,332.92) --
	(106.60,332.92);

\path[draw=drawColor,line width= 0.6pt,line join=round] (103.85,341.05) --
	(106.60,341.05);

\path[draw=drawColor,line width= 0.6pt,line join=round] (103.85,349.17) --
	(106.60,349.17);

\path[draw=drawColor,line width= 0.6pt,line join=round] (103.85,357.30) --
	(106.60,357.30);

\path[draw=drawColor,line width= 0.6pt,line join=round] (103.85,365.43) --
	(106.60,365.43);
\end{scope}
\begin{scope}
\path[clip] (  0.00,  0.00) rectangle (339.67,375.80);
\definecolor{drawColor}{gray}{0.20}

\path[draw=drawColor,line width= 0.6pt,line join=round] (116.95, 24.58) --
	(116.95, 27.33);

\path[draw=drawColor,line width= 0.6pt,line join=round] (168.67, 24.58) --
	(168.67, 27.33);

\path[draw=drawColor,line width= 0.6pt,line join=round] (220.39, 24.58) --
	(220.39, 27.33);

\path[draw=drawColor,line width= 0.6pt,line join=round] (272.11, 24.58) --
	(272.11, 27.33);

\path[draw=drawColor,line width= 0.6pt,line join=round] (323.83, 24.58) --
	(323.83, 27.33);
\end{scope}
\begin{scope}
\path[clip] (  0.00,  0.00) rectangle (339.67,375.80);
\definecolor{drawColor}{RGB}{0,0,0}

\node[text=drawColor,anchor=base,inner sep=0pt, outer sep=0pt, scale=  0.70] at (116.95, 17.56) {0.00};

\node[text=drawColor,anchor=base,inner sep=0pt, outer sep=0pt, scale=  0.70] at (168.67, 17.56) {0.25};

\node[text=drawColor,anchor=base,inner sep=0pt, outer sep=0pt, scale=  0.70] at (220.39, 17.56) {0.50};

\node[text=drawColor,anchor=base,inner sep=0pt, outer sep=0pt, scale=  0.70] at (272.11, 17.56) {0.75};

\node[text=drawColor,anchor=base,inner sep=0pt, outer sep=0pt, scale=  0.70] at (323.83, 17.56) {1.00};
\end{scope}
\begin{scope}
\path[clip] (  0.00,  0.00) rectangle (339.67,375.80);
\definecolor{drawColor}{RGB}{0,0,0}

\node[text=drawColor,anchor=base,inner sep=0pt, outer sep=0pt, scale=  0.90] at (220.39,  7.25) {Relative Error};
\end{scope}
\begin{scope}
\path[clip] (  0.00,  0.00) rectangle (339.67,375.80);
\definecolor{drawColor}{RGB}{0,0,0}

\node[text=drawColor,rotate= 90.00,anchor=base,inner sep=0pt, outer sep=0pt, scale=  0.90] at ( 11.70,198.82) {Network};
\end{scope}
\end{tikzpicture}